\documentclass[
submission
]{dmtcs-episciences}

\usepackage[utf8]{inputenc}
\usepackage{subfigure}

\usepackage[round]{natbib}

\usepackage{xspace}
\usepackage{complexity}
\usepackage{csquotes}
\usepackage{dsfont}
\usepackage[disable]{todonotes}
\usepackage{mdframed}
\usepackage{mathtools}
\usepackage{thmtools,thm-restate}
\usepackage{tikz}
\usepackage{xcolor}
\usepackage[running, left, mathlines]{lineno}
\usepackage[noabbrev,capitalize]{cleveref}

\newcommand{\probname}[1]{{\normalfont\textsc{#1}}}
\newcommand{\probdef}[3]{%
    \begin{mdframed}
		\probname{#1}
		\begin{description}
            \item[Given:] #2
            \item[Question:] #3
        \end{description}
	\end{mdframed}
}

\newcommand{\realizability}{\probname{Planar Linkage Realizability}\xspace}
\newcommand{\realizabilityShort}{\probname{PLR}\xspace}

\newcommand{\instance}{\ensuremath{\mathcal{I}}\xspace}
\newcommand{\instanceLong}{\ensuremath{(\linkage=(G,\ell), P, V', \delta)}\xspace}

\newcommand{\linkage}{\ensuremath{\mathcal{L}}\xspace}
\newcommand{\configuration}{\ensuremath{\Gamma}\xspace}
\newcommand{\Unit}{\ensuremath{\mathds{1}}}

\definecolor{PositiveColor}{RGB}{3,173,56}
\definecolor{NegativeColor}{RGB}{215,25,28}
\NewDocumentCommand{\NewEntrance}{O{black} O{0.5} m m}{%
	\ensuremath{{}^{\rightarrow}\begin{tikzpicture}
			\fill[fill opacity=#2, #1] (0,0) circle[radius=0.9ex];
			\node[inner sep=0pt, font=\tiny] at (0,0) {$#3$};
		\end{tikzpicture}^{}_{#4}}%
}
\NewDocumentCommand{\NewExit}{O{black} O{0.5} m m}{%
	\ensuremath{\begin{tikzpicture}
			\fill[fill opacity=#2, #1] (0,0) circle[radius=0.9ex];
			\node[inner sep=0pt, font=\tiny] at (0,0) {$#3$};
		\end{tikzpicture}^{\rightarrow}_{#4}}%
}
\NewDocumentCommand{\Entrance}{s O{black} O{0.5} m}{%
	\ensuremath{{}^{\rightarrow}\begin{tikzpicture}
			\fill[fill opacity=#3, #2] (0,0) -- (1.5ex,0) arc (0:\IfBooleanT{#1}{-}90:1.5ex) -- cycle;
		\end{tikzpicture}^{}_{#4}}%
}
\NewDocumentCommand{\Exit}{s O{black} O{0.5} m}{%
	\ensuremath{\begin{tikzpicture}
			\fill[fill opacity=#3, #2] (0,0) -- (1.5ex,0) arc (0:\IfBooleanT{#1}{-}90:1.5ex) -- cycle;
		\end{tikzpicture}^{\rightarrow}_{#4}}%
}
\NewDocumentCommand{\Start}{s O{black} O{0.5} m}{%
	\ensuremath{{}^{\rightarrow}\begin{tikzpicture}
			\draw[opacity=#3, #2] (0,0) -- (1.5ex,0) arc (0:\IfBooleanT{#1}{-}90:1.5ex) -- cycle;
			\fill[fill opacity=#3, #2] (0,0) -- (1ex,0) -- (0,\IfBooleanT{#1}{-}1ex) -- cycle;
		\end{tikzpicture}^{}_{#4}}%
}
\NewDocumentCommand{\End}{s O{black} O{0.5} m}{%
	\ensuremath{\begin{tikzpicture}
			\draw[opacity=#3, #2] (0,0) -- (1.5ex,0) arc (0:\IfBooleanT{#1}{-}90:1.5ex) -- cycle;
			\fill[fill opacity=#3, #2] (0,0) -- (1ex,0) -- (0,\IfBooleanT{#1}{-}1ex) -- cycle;
		\end{tikzpicture}^{\rightarrow}_{#4}}%
}
\NewDocumentCommand{\ExampleArea}{O{black} O{0.5}}{%
	\ensuremath{\begin{tikzpicture}
			\fill[fill opacity=#2, #1] (0,1.5ex) rectangle (1.5ex,0);
	\end{tikzpicture}}%
}
\newcommand{\NegativeEntrance}[1]{\NewEntrance[NegativeColor][0.5]{-}{#1}}	
\newcommand{\NegativeExit}[1]{\NewExit[NegativeColor][0.5]{-}{#1}}
\newcommand{\PositiveEntrance}[1]{\NewEntrance[PositiveColor][0.5]{+}{#1}}	
\newcommand{\PositiveExit}[1]{\NewExit[PositiveColor][0.5]{+}{#1}}

\newcommand{\Quadrant}[2]{\ensuremath{Q_{\uppercase\expandafter{\romannumeral#1\relax}}(#2)}}

\newcommand{\ExistsR}{\ensuremath{\exists\mathbb{R}}}
\newcommand{\ETR}{\ExistsR}
\newcommand{\Size}[1]{\ensuremath{\left\vert #1 \right\vert}}
\newcommand{\BigO}[1]{\ensuremath{\mathcal{O}(#1)}}

\newcommand{\NewText}[1]{{#1}}

\graphicspath{{./graphics/}}

\newtheorem{theorem}{Theorem}
\newtheorem{lemma}{Lemma}
\newtheorem{observation}{Observation}

\usepackage{ifthen}
\usepackage{pifont}
\newboolean{long}
\newcommand{\AppendixSymbol}{\ding{72}}
\setboolean{long}{true} %
\ifthenelse{\boolean{long}}{
	\NewDocumentEnvironment{prooflater}{m}{\begin{proof}}{\end{proof}\ignorespacesafterend}
	\NewDocumentEnvironment{proofsketch}{o +b}{}{\ignorespacesafterend}
	\newcommand{\restateref}[1]{}
	\NewDocumentEnvironment{statelater}{m}{}{}

	\NewDocumentCommand{\onlyShort}{+m}{}
	\NewDocumentCommand{\onlyLong}{+m}{#1}
	\NewDocumentCommand{\shortLong}{+m +m}{#2}
}{
	\NewDocumentEnvironment{prooflater}{m +b}{%
		\expandafter\global\expandafter\def\csname#1\endcsname{\begin{proof}#2\end{proof}}%
	}{\ignorespacesafterend}
	\NewDocumentEnvironment{proofsketch}{O{Proof sketch}}{\begin{proof}[#1]}{\end{proof}\ignorespacesafterend}
	\usepackage{apptools}

	\newcommand{\restateref}[1]{[\IfAppendix{\AppendixSymbol{}}{\AppendixSymbol{}}]}
	
	\NewDocumentEnvironment{statelater}{m +b}{%
		\expandafter\global\expandafter\def\csname#1\endcsname{#2}%
	}{\ignorespacesafterend}

	\NewDocumentCommand{\onlyShort}{+m}{#1}
	\NewDocumentCommand{\onlyLong}{+m}{}
	\NewDocumentCommand{\shortLong}{+m +m}{#1}
}

\let\oldrestatable\restatable
\def\restatable{\expandafter\oldrestatable}

\author[T. Depian, C. Haase, M. Nöllenburg, A. Schulz]{Thomas Depian\affiliationmark{1}\thanks{Supported by the Vienna Science and Technology Fund (WWTF) [10.47379/ICT22029].}
  \and Carolina Haase\affiliationmark{2}
  \and Martin Nöllenburg\affiliationmark{1}\thanks{Supported by the Vienna Science and Technology Fund (WWTF) [10.47379/ICT22029].}
  \and Andr{\'e}~Schulz\affiliationmark{3}
  }
\title{Realizing Planar Linkages in Polygonal Domains}
\affiliation{
  Algorithms and Complexity Group, TU Wien, Austria\\
  Trier University, Germany\\
  FernUniversität in Hagen, Germany}
\begin{document}
\maketitle
\begin{abstract}
 A linkage $\linkage$ consists of a graph $G=(V,E)$ and an edge-length function~$\ell$.
Deciding whether \linkage can be realized as a planar straight-line embedding in $\mathbb R^2$ with edge length $\ell(e)$ for all $e \in E$ is {%
$\ExistsR$}-complete [Abel et al., JoCG'25], even if $\ell \equiv 1$, but a considerable part of $\linkage$ is rigid.

In this paper, we study the computational complexity of the realization question for structurally simpler, less rigid linkages inside an open polygonal domain $P$, where the placement of some vertices may be specified in the input.
We show \NewText{that the problem is in \XP\ parameterized by the size of $G$, and} %
\W[1]-hard %
even if $\ell \equiv 1$ and no vertex positions are constrained. %
Furthermore, we consider the case where $G$ is a path with constrained start and end position and $\ell\equiv 1$.
Despite the absence of any rigid components, we obtain \NP-hardness in general, and provide a $\BigO{n \log n}$ time algorithm for $n$ being the number of vertices of $P$ for arbitrary~$\ell$ if $G$ has only three edges\NewText{, both endpoints are constrained,} and $P$ is convex.
\end{abstract}

\section{Introduction}

In this paper, we study \emph{linkages} $\linkage = (G, \ell)$, which are graphs $G = (V,E)$  where every edge $e \in E$ has a predetermined length $\ell(e)$.
Linkages can be used to model mechanical constructions, e.g., robot arms consisting of metal bars connected at joints, %
or folding proteins \citep{DO.GFA.2007}, and their study is well-established in computational geometry and mechanical
engineering: famous results such as Kempe's universality theorem date back to the 1870s \citep[see][]{Kem.GMd.1875,KM.Utc.2002}. By now, there exists a plethora of work on linkages \citep[see also][]{HJW.MP2.1984,ADD+.WNC.2016}.
A realization of a linkage (an embedding of the vertices $V$ such that edge lengths given by $\ell$ are realized) is called a \emph{configuration}. A configuration in the plane is called \emph{planar} if no two edges cross. In this work, we focus  on planar configurations. Planarity is often a requirement given by the application. For example,
bars of a mechanical framework might not have the option to overlap and hence have to stay disjoint.
We refer to the book by 
\citet{DO.GFA.2007} for an overview of existing results and applications; for both planar and non-planar configurations.

\realizability (\realizabilityShort for short) is the problem that asks whether a linkage has a planar realization,  i.e., deciding whether a linkage \linkage admits a planar configuration.
It is known to be complete for the \emph{Existential Theory of the Reals} ({%
\ExistsR}) \citep{ADD+.WNC.2016,ADD+.Wnc.2025}, even if \linkage is \emph{unit-length}, i.e., all edge lengths are the same (also known as \emph{matchstick graphs}).
Known hardness constructions rely on many rigid components, whose %
\NewText{embeddings are unique} up to rigid motions \citep[examples are][]{EW.Fel.1990,CDR.PEG.2004,CDR.PEG.2007,ADD+.Wnc.2025}.
This raises the question whether these hardness results carry over to more \enquote{flexible} linkages.
The related \emph{reconfiguration} problem, i.e., deciding whether there exists a continuous, edge-length-preserving, and sometimes also planar, motion from one configuration into another, is \PSPACE-hard \citep{CDR.SPA.2003,Kan.Msl.1992,WP.RCE.1996,BDD+.nrt.2002,AKRW.cu.2003}. This holds in many variants, but 
notably also if $\mathcal{L}$ is a \emph{linear linkage},
that is, $G$ is a path, in the presence of obstacles \citep{HJW.MP2.1984,JP.CRM.1985}.
The problem remains \NP-hard for obstacles made up of four axis-aligned segments as shown by \citet{HJW.MRA.1985}.
Note that the above-mentioned reductions require non-unit edge lengths.

Motivated by the lack of analogous results for \realizabilityShort, we investigate in this paper the realizability question for structurally simple linkages~\linkage in the presence of obstacles, modeled via an open polygonal domain~$P$ into which we must embed~$\linkage$ in a planar way.
This natural variant of the realization problem has, to the best of our knowledge, not been studied before.
{%
Note that related problems, e.g., finding a planar straight-line drawing of a graph inside a polygonal region with some constrained vertex positions, do not have edge-length constraints \citep{LMM.CDG.2018,LMM.CDG.2022}.
}

\paragraph{Contributions.}
We show in \Cref{sec:xp-membership} that \realizabilityShort is in \XP\footnote{We assume familiarity with basic concepts of \emph{parameterized complexity}%
 and refer to the book by \citet{Cygan2015} for an introduction.} with respect to the size of the graph (\Cref{thm:xp})
by providing an \ETR-formulation of our problem.
\NewText{Such formulas} %
 are known to be efficiently solvable for a bounded number of variables \citep{GJ.SSP.1988}.
\NewText{Since} \realizabilityShort is {%
	 \ExistsR}-complete even without polygonal domain~$P$, %
	 it is no surprise that it can be expressed as an \ETR-formula.
	 \NewText{However, we underline in }\Cref{sec:w1-hardness} %
	 that this upper bound can be considered ``best possible''.
More formally, we show that \realizabilityShort is \W[1]-hard when parameterized by the size of $G$ even for unit-length linkages (\Cref{thm:w1}).
Apart from ``justifying'' \Cref{thm:xp}, this also shows that no structural parameter of $G$, not even the number of rigid components, can lead to an \FPT-algorithm for \realizabilityShort.

{%
We then focus in \Cref{sec:linear-linkages,sec:three-edge-linkages} on linear linkages, i.e., where $G$ is a path. %
If the linkage is also unit-length, i.e., $\ell \equiv 1$, then \realizabilityShort is equivalent to finding a placement for a single edge as we can place the entire linkage arbitrarily close to it \NewText{(if $V' = \emptyset$)}. 
However, if we fix the position of the endpoints of the path, the problem is less trivial.
Towards understanding the complexity of this special case of \realizabilityShort, we show, on the one hand, that it is \NP-hard in general (\Cref{thm:hardness}).
On the other hand, if the polygonal domain~$P$ is convex and $G$ has three edges, then we can solve \realizabilityShort in time $\BigO{n \log n}$ for $n$ being the number of vertices of $P$, even for arbitrary $\ell$ (\Cref{thm:linear-three-edges}).
The algorithm exploits the observation that three-edge linear linkages \NewText{with fixed endpoints} have only one degree of freedom, which allows us to examine a linear number of candidate configurations.
\Cref{thm:hardness,thm:linear-three-edges} together suggest that the complexity of \realizabilityShort in this surprisingly challenging setting arises from an interplay between the degrees of freedom in $\linkage$ and the obstacles in $P$, and we see these two results as a first step towards closing this complexity gap.%
}

\section{Preliminaries}
\label{sec:preliminaries}
A \emph{linkage} $\linkage = (G, \ell)$ consists of a simple, undirected, and connected graph $G$ with vertex set $V(G)$ and edge set $E(G)$ and a function $\ell\colon E(G)\to \mathbb{R}_{>0}$ assigning each edge $e \in E(G)$ a positive \emph{length} $\ell(e)$.
The linkage \linkage is \emph{linear} if $G$ is a path and \emph{unit-length} if %
$\ell \equiv 1$, %
abbreviated as $(G, \Unit)$.
A \emph{configuration}~$\configuration$ of a linkage~$\linkage$ is a %
straight-line drawing of $G$, %
such that we have $\Vert \configuration(u) - \configuration(v) \Vert_2 = \ell(uv)$ for every $uv \in E(G)$; we say $\configuration$ \emph{respects}~$\ell$.
\NewText{We call \configuration planar if it is crossing-free, i.e, edges intersect only at common endpoints.
For a point $p \in \mathbb{R}^2$, we use $x(p)$ and $y(p)$ to address its $x$- and $y$-coordinate, respectively.}
Let $P$ be an open polygonal domain, i.e., %
\NewText{a bounded, open, connected set $P \subset \mathbb{R}^2$ possibly with holes defined by a set of polygons.}
We say that \configuration \emph{lives} inside $P$ if %
\NewText{every vertex and edge of $G$ is contained in $P$}. %
In this paper, we allow the placement of some vertices $V' \subseteq V(G)$ to be \emph{constrained} by a function $\delta$ and ask for the existence of a realization of $\linkage$ that \emph{adheres} to $\delta$ within a polygonal domain~$P$:
\probdef{\realizability~(\realizabilityShort)}{A linkage $\linkage = (G, \ell)$, an open polygonal domain $P \subset \mathbb{R}^2$, and a mapping $\delta\colon V'\to P$ for some $V' \subseteq V(G)$.}{Does there exist a planar configuration $\configuration$ of \linkage that lives inside~$P$ such that $\configuration(v) = \delta(v)$ for every $v \in V'$?}

Let $\instance=\instanceLong$ be an instance of \realizabilityShort.
We denote by $n_G \coloneqq \Size{V(G)}$ and $n_P \coloneqq \Size{V(P)}$ the number of vertices of $G$ and $P$, respectively. %
We assume the \emph{Real RAM} model of computation for our algorithmic results; see also \citet{EvdHM.SGN.2024,PS.CGI.1985}.
{%
We consider only rational edge lengths $\ell(\cdot)$, \NewText{constrained vertex placements $\delta(\cdot)$}, and polygonal domains $P$ with rational vertex coordinates \NewText{of bounded bit-complexity}.
}

\section{\textsc{\NewText{Planar} Linkage Realizability} is in \textsf{XP}}
\label{sec:xp-membership}
In this section, we show that \realizabilityShort is in \XP\ with respect to the number of vertices $n_G = \Size{V(G)}$ %
\NewText{for arbitrary $\ell$}; \Cref{sec:w1-hardness} provides a complementing \W[1]-hardness result.
We devise an Existential Theory of the Reals (\ETR) formulation $\varphi$ using only $2n_G$ variables that expresses \realizabilityShort.
It is known that checking whether $\varphi$ is feasible is in \XP\ with respect to its number of variables \citep{GJ.SSP.1988}.
Recall that deciding if a linkage admits a planar configuration is {\ExistsR}-complete \citep{ADD+.Wnc.2025}.
It is {%
not hard to see} that %
\NewText{also the variant with $P$ and $\delta$ lies in $\ExistsR$}. %
\shortLong{
However, to the best of our knowledge, the literature does not contain an explicit formula for \realizabilityShort or a related problem that does not rely on some general position assumption that we cannot make.
{%
Since the running time of our \XP-algorithm depends on the number of variables, (in)equalities and their degree in the \ETR-formula, we sketch below an explicit formula~$\varphi(\instance)$ for an instance \instance of \realizabilityShort; details can be found in the full version \citep{ARXIV}.
}

\subsubsection*{A Sketch of $\boldsymbol{\varphi(\instance)}$.}
We introduce two existentially-quantified variables $x_v$ and~$y_v$ for every vertex $v \in V(G)$ to encode the position of the vertex $v$ in a configuration \configuration, i.e., $\configuration(v) = (x_v, y_v)$.
Moreover, $\varphi(\instance)$ contains four subformulas, ensuring that (1) \configuration adheres to $\delta$, (2) \configuration respects the length constraint~$\ell$, (3)~\configuration is planar, and (4) \configuration lives inside $P$, respectively.
Subformulas~(1) and (2) can be realized by forcing the value of some variables and by encoding the definition of the (squared) Euclidean distance, respectively.
For~(3), we can use the sign of the determinant 
$\Delta(p,q,r) = 
\begin{vsmallmatrix}
x_p & y_p & 1\\
x_q & y_q & 1\\
x_r & y_r & 1 
\end{vsmallmatrix}$ 
to check if the point~$r$ is to the left, to the right, or on the line $\overline{pq}$ through the points $p$ and~$q$ \citep{dBCvKO.CGA.2008}:
If $p$, $q$ are on different sides of $\overline{uv}$ and $u$, $v$ are on different sides of $\overline{pq}$, then the edges $uv, pq \in E(G)$ cross.
If three vertices are collinear, the existence of a crossing depends on the concrete position of the vertices.
Therefore, we test in this case if the $x$- and $y$-coordinate of a vertex, say $p$, is between those of $u$ and $v$, respectively.
The edges cross if this holds for at least one vertex.
For~(4), we first use $\Delta(\cdot)$ to force one vertex to lie inside $P$ (which is already the case if there exists a vertex $v \in V'$).
Afterwards, we ensure that the combined drawing of $G$ and $P$ is planar.
}{
However, to the best of our knowledge, the literature does not contain an explicit formula for \realizabilityShort or a related problem that does not rely on some general position assumption that we cannot make.
Since the running time of our \XP-algorithm depends on the number of variables, (in)equalities and their degree in the \ETR-formula, we provide in the following an explicit formula.

}

\begin{statelater}{etrformula}
Let $\instance = \instanceLong$ be an instance of \realizabilityShort.
The \ETR\ formula $\varphi(\instance)$ contains two variables $x_v$ and $y_v$ for every vertex $v \in V(G)$.
These variables encode the position of the vertex $v$ in a (hypothetical) configuration \configuration, i.e., $\configuration(v) = (x_v, y_v)$.
Let $V(G)= \{v_1, \ldots, v_{n_G}\}$, we set
\begin{align*}
	\varphi(\instance) = \exists x_{v_1},y_{v_1},\ldots,x_{v_{n_G}},y_{v_{n_G}}\colon\bigwedge_{i=1,2,3,4}\varphi_i(\instance),
\end{align*}
where each subformula $\varphi_i(\instance)$ is responsible for ensuring a specific property of \configuration, i.e., that (1) \configuration adheres to $\delta$, (2) \configuration respects the length constraint $\ell$, (3) \configuration is planar, and (4) \configuration lives inside $P$, respectively.
In the following, we discuss each subformula in greater detail.
\medskip

\noindent\textit{The Subformulas of ${\varphi(\instance)}$.}
The first subformula can be defined as follows:
\begin{align*}
	\varphi_1(\instance) = \bigwedge_{v \in V'} (x_v = x(\delta(v)) \land y_v = y(\delta(v)))
\end{align*}

For the second subformula, we use the definition of the \NewText{(squared)} Euclidean distance:
\begin{align*}
	\varphi_2(\instance) = \bigwedge_{uv \in E(G)} ((x_v-x_u)(x_v - x_u) + (y_v - y_u)(y_v-y_u) = \ell(uv)\cdot\ell(uv))
\end{align*}

With the third subformula, $\varphi_3(\instance)$, we aim to ensure planarity.
To this end, we use the well-known characterization of crossings that exploits the sign of a carefully constructed determinant; for further information see the book by \citet{dBCvKO.CGA.2008}.
More formally, let $u,v,p\in V(G)$ be three vertices.
We let $\Delta(u,v,p)$ denote the determinant 
\begin{align*}
	\begin{vmatrix}
		x_u & y_u & 1\\
		x_v & y_v & 1\\
		x_p & y_p & 1
	\end{vmatrix},
\end{align*}
which is positive, zero, \NewText{or} negative if $p$ lies to the left of, on, \NewText{or} to the right of the \NewText{directed} line \NewText{from} $u$ \NewText{to} $v$, respectively.
Since we cannot enforce general position for \configuration, particular care must be taken if three vertices are collinear, i.e., if $\Delta(u,v,p) = 0$ holds.
To facilitate the later definition of $\varphi_3(\instance)$, we first define some auxiliary subformulas.
Let $u,v,p \in V(G)$ be three vertices with $uv \in E(G)$.
The formula $\text{onEdge}(u,v,p)$ determines if the vertex $p$ lies on the edge $uv$.
Note that this induces a crossing and is, therefore, not allowed in \configuration.
\begin{align*}
    \text{between}(u,v,p) &= (x_v-x_p)(x_u-x_p) \leq 0 \land (y_v-y_p)(y_u-y_p) \leq 0\\
	\text{onEdge}(u,v,p) &= \Delta(u,v,p) = 0 \land \text{between}(u,v,p)
\end{align*}
Note that the first term of $\text{onEdge}(u,v,p)$ determines if $p$ is collinear with %
$u$ and $v$ and $\text{between}(u,v,p)$ checks if the $x$- and $y$-coordinate of $p$ lie between the $x$- and $y$-coordinates of $u$ and $v$, respectively.
Observe that if all three %
\NewText{(in)equalities} hold, then $p$ must be on the edge $uv$.
\NewText{Note that this also captures the case where two vertices are placed on the same point.}

Next, let $uv,pq \in E(G)$ be two independent edges, i.e., $\Size{\{u,v,p,q\}} = 4$.
We now present the formula $\text{intersect}_{\text{ind}}(uv,pq)$, which checks if these edges intersect:
\begin{align*}
	\text{intersect}_{\text{ind}}(uv,pq) &= (\Delta(u,v,p)\Delta(u,v,q) < 0 \land \Delta(p,q,u)\Delta(p,q,v) < 0)\\
	&\lor \text{onEdge}(u,v,p) \lor \text{onEdge}(u,v,q)\\
	&\lor \text{onEdge}(p,q,u) \lor \text{onEdge}(p,q,v).
\end{align*}
Similarly, let $uv, uw \in E(G)$ be two edges that share one endpoint, in this case the vertex $u \in V(G)$.
For this type of edge pairs, we use the formula $\text{intersect}_{\text{dep}}(uv,uw)$ to check for crossings:
\begin{align*}
	\text{intersect}_{\text{dep}}(uv,uw) &= \text{onEdge}(u,v,w) \lor \text{onEdge}(u,w,v).
\end{align*}

Finally, we are able to define the subformula $\varphi_3(\instance)$:
\begin{align*}
	\varphi_3(\instance) &= \left(\bigwedge_{\substack{uv, pq \in E(G),\\\Size{\{u,v,p,q\}}=4}}\lnot \text{intersect}_{\text{ind}}(uv,pq)\right)\\
    &\land \left(\bigwedge_{\substack{uv, uw \in E(G),\\v\neq w}} \lnot \text{intersect}_{\text{dep}}(uv,uw)\right).
\end{align*}

With the last subformula, i.e., $\varphi_4(\instance)$, we ensure that the configuration \configuration lives inside the polygonal domain $P$.
If $V' \neq \emptyset$, i.e., the input constrains the position of at least one vertex of $G$, then $\varphi_1(\instance)$ forces one vertex of $G$ to lie inside $P$.
If $V' = \emptyset$, we need an additional subformula that ensures that a vertex of $G$, say $v\in V(G)$, lies inside $P$.
\NewText{This can be enforced by, for example, triangulating $P$ in polynomial time and requiring that $v$ lies in (or on) one of the triangles (and not on the boundary of $P$), which can be tested using the determinant from $\varphi_3$.}

Once we know that at least one vertex of $G$ lies inside $P$, it suffices to ensure that $\configuration$ never leaves $P$.
This can be enforced by checking if $G$ and $P$ are crossing free, i.e., no edge of $G$ should cross an edge of $P$, \NewText{and no vertex is on an edge of $P$.}
We can encode this via an adapted version of $\varphi_3(\instance)$.
If the test succeeds, then $\configuration$ lives inside $P$ as no edge crosses the boundary of $P$; recall that $P$ is an open polygonal domain.

\end{statelater}
\begin{theorem}
	\label{thm:xp}
	\realizability\ is in \ExistsR\ and in \XP\ with respect to the number $n_G$ of vertices in $G$.
\end{theorem}
\begin{proof}
    Let \instance be an instance of \realizabilityShort.
    Containment in \ExistsR\ follows from the existence of $\varphi(\instance)$\NewText{, which can be computed in polynomial time.}
    To show \XP-membership, we analyze the formula.
    It has $2n_G$ variables and $\BigO{{n_G}^2 + n_G \cdot n_P}$ equalities and inequalities; note that $\Size{E(G)} \in \BigO{n_G}$ holds as \NewText{we can reject instances where $G$ is not planar, which can be checked in polynomial time.} %
	It is known that feasibility of an \ETR\ formula on $N$ variables and $M$ (in)equalities can be checked in $L^{\BigO{1}}(M\cdot D)^{\BigO{N^2}}$ time, where~$D$ is the maximum degree of the (in)equalities and $L$ denotes the bit-complexity of their coefficients \citep{GJ.SSP.1988}.
	In our case, $L$ %
	is polynomial in the size of \instance\ \NewText{and $D$ is constant.} 
	\NewText{Hence, feasibility of $\varphi(\instance)$ can be decided in $L^{\BigO{1}}(n_G + n_P)^{\BigO{{n_G}^2}}$ time.}
\end{proof}

\section{Unit-Length Linkage Realizability is \textsf{W}[1]-hard}
\label{sec:w1-hardness}
We next complement \Cref{thm:xp} with a corresponding \W[1]-hardness result.
To show \W\textup{[1]}-hardness, we use a reduction from the \W\textup{[1]}-hard problem \probname{Grid Tiling} \citep{Cygan2015}. Here, given two integers~$k$ and $m$, as well as a collection $S$ of size $k^2$ containing nonempty sets $S_{i,j} \subseteq [m] \times [m]$ with $1 \leq i,j \leq k$, the goal is to find for each of the sets a tuple $s_{i,j} \in S_{i,j}$, such that for two such tuples $(a,b) = s_{i,j}$ and $(a',b') = s_{i',j'}$ if $i = i' + 1$ and $j = j'$ ($j = j' + 1 $ and $i = i'$), then %
\NewText{$b = b'$ ($a = a'$)}; see \cref{fig:grid-tiling}c.

We construct a unit-length linkage $(G, \Unit)$ by subdividing each edge of a $k \times k$ grid graph and %
\NewText{attaching} to each vertex of the original grid two adjacent vertices to create an (equilateral) triangle. We refer to the vertices of the original grid graph as \emph{grid vertices}, the vertices used to subdivide the edges as \emph{subdivision vertices} and all other vertices as \emph{triangle vertices}.
We define a polygonal domain~$P'$ as a $k \times k$ grid, where every edge has width $w$. We refer to the square shaped regions in $P'$ that correspond to the grid points as \emph{grid squares} \NewText{$R_{i,j}$, $i,j \in [k]$}; see \cref{fig:grid-tiling}a.
The goal is for each grid vertex of $G$ to be placed in one of the grid squares and for no two grid vertices to be placed in the same grid square.
In each grid square we define $m^2$ points on an $m \times m$ grid, such that every point represents a value $(a,b) \in [m] \times [m]$. We want each grid vertex to be placed on one of those points or in a small $\varepsilon$-region around it\NewText{, for a suitable $\varepsilon$ that we define later}. If a grid vertex is placed in such an $\varepsilon$-region, we say it \emph{uses} the corresponding point.
For each grid square $R_{i,j}$ we need to ensure that only points $(a,b)$ with $(a,b) \in S_{i,j}$ can be used. To this end, we add for each value $(a,b) \in S_{i,j}$ a so-called triangle corridor to $P'$.  Placing a grid vertex in an $\varepsilon$-region then corresponds to choosing a tuple in the \probname{Grid Tiling} instance.

\begin{figure}
\centering
\includegraphics[page=1]{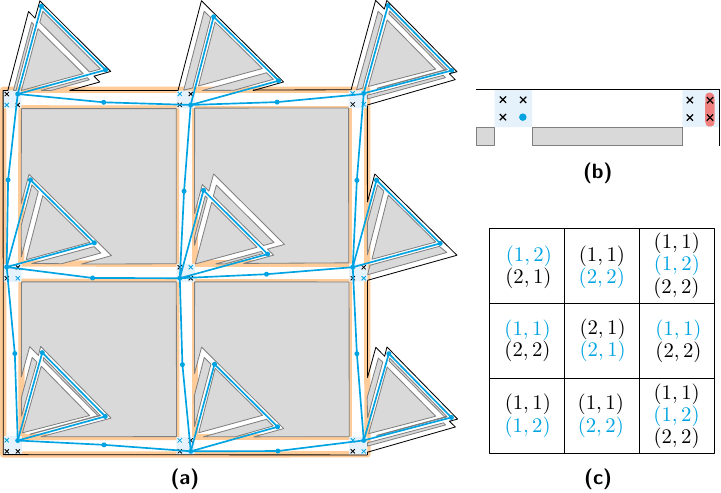}
\caption{\textbf{\textsf{(a)}} Construction for $k=3$ and $m=2$, $P'$ highlighted in orange, grid squares marked in light blue. \textbf{\textsf{(b)}} If the blue point with tuple $(2,1)$ is chosen, then to the right only points in the red area with tuples %
	\NewText{$(2,x)$}, $1\leq x \leq m$ can be used. \textbf{\textsf{(c)}} \probname{Grid Tiling} instance of the construction in (a).}
\label{fig:grid-tiling}
\end{figure}

\begin{restatable}\restateref{le:triangles}{lemma}{lemmaTriangles}
	\label{le:triangles}
	If a grid vertex \NewText{$g$} lies inside a grid square $R_{i,j}$, then the adjacent triangle has to lie inside a triangle corridor and \NewText{$g$} uses a point $(a,b) \in S_{i,j}$.
\end{restatable}
\begin{prooflater}{plemmaTriangles}
Let $g$ be a grid vertex and $t_1$, $t_2$ its adjacent triangle vertices.
A triangle cannot be placed outside of a triangle corridor inside $P'$. We place the triangle corridor in such a way that, if $g$ is placed directly at $(a,b) \in S_{i,j}$, the triangle fits into the triangle corridor. Both the inner and outer boundary of the triangle corridor are parts of concentric equilateral triangles.
If $g$ is not placed directly at $(a,b)$ we still need to ensure visibility between $g$, $t_1$ and $t_2$. This allows us to determine small areas $T_1$ and $T_2$ in the corners of the triangle corridor, in which $t_1$ and $t_2$ must be placed; see \cref{fig:triangle-width}a. Hence, the size of $T_1$ and $T_2$ determines how far $g$ can be moved away from $(a,b)$.

Let $\varepsilon$ be the width of the triangle corridor.
Since both boundaries of the triangle corridor are (parts of) equilateral triangles, the two corner points $c_1$ and $c_2$ have a distance of $2\varepsilon$, which is the largest distance within $T_1$ and $T_2$; see \cref{fig:triangle-width}b.  
Thus, we can move $g$ by at most $2\varepsilon$ in any direction.

To prevent $g$ from using points $(a',b') \notin S_{i,j}$, i.e. points that do not have a triangle corridor, we need to choose $\varepsilon$ small enough, such that the distance between points $(a,b) \in R_{i,j}$ is \NewText{larger} %
than $4\varepsilon$. We set the distance between the points to %
 $w / m$, where $w$ is the width of the edges in $P'$. Thus, we need $4\varepsilon < w / m$. %
\begin{figure}
\centering
\includegraphics[page=3]{grid-tiling}
\caption{\textbf{\textsf{(a)}} The movement of $g$ is restricted by the movement of $t_1, t_2$; \textbf{\textsf{(b)}} the longest distance in the areas $T_1$ and $T_2$ is $2\varepsilon$.}
\label{fig:triangle-width}
\end{figure}
\end{prooflater}

Note also that no two grid vertices can be in the same grid square, as their triangles would cross\NewText{; see also \Cref{fig:grid-tiling}.}
Lastly, for our construction to work as intended, we have to choose the width $w$ of the edges of $P'$ such that if a point $(a,b)$ is used by a grid vertex $g_{i,j}$ in $R_{i,j}$ then a grid vertex $g_{i,j+1}$ in $R_{i,j+1}$ can use all points $(a,b') \in S_{i,j+1}$, but no other points, and a grid vertex $g_{i+1,j}$ in $R_{i+1,j}$ can use all points %
\NewText{$(a',b) \in S_{i+1,j}$}, but no other points; see \cref{fig:grid-tiling}b.  

\begin{restatable}\restateref{le:gridwidth}{lemma}{lemmaGridWidth}
	\label{le:gridwidth}
	If a grid vertex $g_{i,j}$ uses $(a,b)$ in $R_{i,j}$ then $g_{i,j+1}$ uses $(a,b')$ in $R_{i,j+1}$ \NewText{for some $b'$} and $g_{i+1,j}$ uses $(a', b)$ in $R_{i+1,j}$ \NewText{for some $a'$}.
\end{restatable}
\begin{prooflater}{plemmaGridWidth}
We calculate the width $w$ of the edges of $P'$, such that, if a grid vertex $g_{i,j}$ is placed at $(a,b)$ in $R_{i,j}$ the distance that a grid vertex %
\NewText{$g_{i+1,j}$} can have to the line through the points $(a',b)$ in $R_{i+1,j}$ is (significantly) smaller than the distance between the points. We focus only on this case, as the case with points $(a,b')$ in $R_{i, j+1}$ works \NewText{analogously.} %
We set the horizontal and vertical distance between points $(a,b) \in R_{i,j}$ to $w / m$. %

Consider the extreme case, where vertices can lie on the boundary of the polygonal domain. Then, if $g_{i,j}$ is placed on a point $(a,b)$ on the boundary of $R_{i,j}$, the distance to a point $(a',b)$ in $R_{i+1,j}$ placed on the opposite boundary must be %
\NewText{(at most)}
$2$. Otherwise not all points $(a',b)$ in $R_{i+1,j}$ can be reached exactly. 
Let \NewText{$d_1$} %
be the distance between points $(a,b) \in R_{i,j}$ and $(a,b) \in R_{i+1,j}$ and \NewText{$d_2$} %
the shortest possible distance between two grid vertices $g_{i,j}$ and $g_{i+1,j}$; see \cref{fig:corridor-width}.
We can calculate $d_1$ and $d_2$ using the following equations: $d_2 = \sqrt{4(1-w^2)}$ and $d_1 = \sqrt{4 - w^2}$. 
The longest vertical distance that $g_{i+1,j}$ can be away from a point $(a',b)$ in $R_{i+1,j}$ is $d_1-d_2$. Hence, we need $d_1 - d_2 = \sqrt{4 - w^2} - \sqrt{4(1-w^2)} < w/2m$, %
which we can solve for $w$.
\begin{figure}
\centering
\includegraphics[page=2]{grid-tiling}
\caption{Calculat\NewText{ing the} width of edges of $P'$.}
\label{fig:corridor-width}
\end{figure}
\end{prooflater}

The theorem follows from \Cref{le:triangles,le:gridwidth}:

\begin{theorem}
	\label{thm:w1}
	\realizability\ is \W\textup{[1]}-hard with respect to the number $n_G$ of vertices of $G$, \NewText{even for unit-length linkages without constrained vertices}.
\end{theorem}

\section{Realizing Linear Unit-Length Linkages is \NP-hard}
\label{sec:linear-linkages}
\Cref{thm:w1} rules out fixed parameter tractable algorithms for all common graph parameters\NewText{, unless $\FPT=\W[1]$}.
However, similar to previous reductions \citep[see for example][]{ADD+.Wnc.2025,CDR.PEG.2007}, the constructed graph~$G$ still contains rigid components.
In this section, we consider linear linkages, i.e., where $G$ is a path, which do not have any rigid structure.
We show that \realizabilityShort remains \NP-hard even if \linkage is a unit-length linear linkage where we constrain only the placement of the first and last vertex of $G$, i.e., its \emph{endpoints}; see also \Cref{thm:hardness}.

In the following, we first %
\NewText{describe a useful building block in \Cref{sec:hardness-overview}.}
In \Cref{sec:hardness-details}, we provide the detailed geometric construction of our gadgets.
Finally, we combine the gadgets in \Cref{sec:hardness-combine} to establish \Cref{thm:hardness}.

\newcommand{\stateobservationClauseGadget}{
\begin{restatable}{observation}{observationClauseGadget}
\label{obs:clause-gadget-planar}
Let $\configuration$ be a planar configuration of $\linkage$ that enters the main part $C'$ of a clause gadget three times.
For any variable $x_i$, truth assignment $t \in \{0,1\}$ to $x_i$, and vertex $v \in V(G)$, we have that $\Gamma(v) \in \Entrance{i,t}(C')$ implies that there is some $v' \in V(G)$ with $\Gamma(v') \in \Exit{i,t}(C')$.
Furthermore, there is some $v''\in V(G)$ such that $\configuration(v'') \in \Entrance{z,1}(C')$ for some $z \in \{i,j,k\}$.
\end{restatable}
}
\newcommand{\showClauseGadgetOverviewFigure}{
\begin{figure}
\centering
\includegraphics[page=10]{clause-gadget}
\caption{\textbf{\textsf{(a)}}--\textbf{\textsf{(d)}} Different configurations \configuration through the main part of the clause gadget. Dashed lines indicate different possibilities for the configuration to pass through the corridors.}
\label{fig:clause-gadget}
\end{figure}
}
\newcommand{\showTurningALineFigure}{
\begin{figure}
\centering
\includegraphics{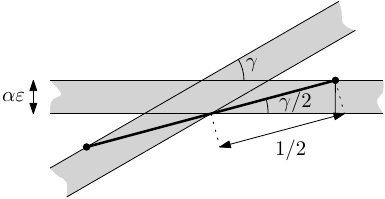}
\caption{A segment of length one can only have endpoints in two corridors
	if $\gamma \le 2 \arcsin{2\alpha \varepsilon}$.}
\label{fig:turning-a-line}
\end{figure}
}

\subsection{Overview of the Reduction and $\boldsymbol{\varepsilon}$-thick Corridors}
\label{sec:hardness-overview}
We give a reduction from the \NP-complete problem \probname{Planar Monotone 3-Sat} \citep{DBK.OBS.2012}.
An instance $\varphi = (\mathcal{X}, \mathcal{C}, \mathcal{E})$ of this problem consists of $N$ variables $\mathcal{X} = \{x_1, \ldots, x_N\}$ and~$M$ clauses $\mathcal{C} = \{c_1, \ldots, c_M\}$, each on (at most) three literals, partitioned into the positive clauses $\mathcal{C}^+$ containing only non-negated literals and negative clauses $\mathcal{C}^-$ containing only negated
literals. 
Moreover, there must exist a planar rectilinear drawing $\mathcal{E}$ of the clause-variable incidence graph of $\varphi$ where all variables are on a horizontal line, which separates $\mathcal{C}^+$ (above the line) and $\mathcal{C}^-$ (below the line); see \Cref{fig:example}a for an example. 
Without loss of generality, we can assume that every clause $c_j$ consists of three (not necessarily distinct) literals.
If a clause does not contain three literals, we can duplicate one of its existing literals.
Note that this does not affect the existence of a planar drawing $\mathcal{E}$.
Moreover, we can %
assume that the drawing~$\mathcal{E}$ is given as part of the input and uses only %
polynomial coordinates %
\citep[see also][]{DBK.OBS.2012,CDR.PEG.2004,CDR.PEG.2007}.

\begin{figure}
\centering
\includegraphics[page=1]{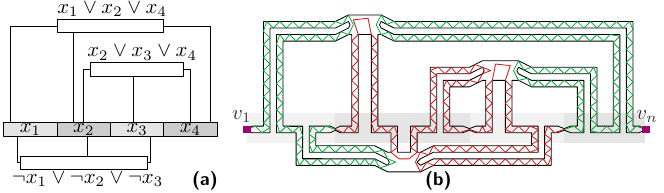}
\caption{\textbf{\textsf{(a)}} A formula $\varphi$ and \textbf{\textsf{(b)}} a schematization of the constructed instance.}
\label{fig:example}
\end{figure}

In our reduction, we create a unit-length linear linkage $\linkage = (G, \Unit)$ on $n$ vertices, where we specify~$n$ in the end.
Conceptually, our construction will force a configuration \configuration of \linkage to visit each variable of $\varphi$ and each clause it is contained in, set truth assignments for the former and verify their truth status for the latter components.  
The path that \configuration follows will closely resemble the drawing $\mathcal{E}$ of $\varphi$, see also \Cref{fig:example}b, and we will replace each edge and vertex of~$\mathcal{E}$ with a gadget, which is a part of the polygonal domain~$P$.
We highlight in each gadget  dedicated regions in the plane, in the following called \emph{areas}, where \configuration must pass through.
More concretely, once the configuration enters a gadget through an entrance area, the construction ensures that it must leave the gadget %
through the corresponding exit area, as it otherwise cannot reach the, yet to be specified, prescribed endpoint $t$. %
The gadgets that we describe in \Cref{sec:hardness-details} are composed of $\varepsilon$-thick corridors, for some small value $\varepsilon$, which we describe next.
To this end, let $\varepsilon \ll 1$ be a small rational constant whose value we fix in \Cref{sec:hardness-combine}.

\paragraph*{Thick Corridors.}
Let $uv$ be an edge with a configuration $\configuration$ (in the plane) such that %
$\Vert \configuration(u) - \configuration(v) \Vert_2 = 1$, and let $n(uv)$ be the normal of $uv$ (of length one).
The \emph{$\varepsilon$-rectangle} $R(uv, \configuration, \varepsilon)$ of $uv$ is the rectangle obtained by translating $uv$ by $\pm \varepsilon/2 n(uv)$ and connecting the endpoints, i.e., $R(uv, \configuration, \varepsilon)$ is a rectangle of width $\varepsilon$ and length $1$ such that $uv$ connects the midpoints of its two shorter sides.
Let $\mathcal{Q}$ be a linear unit-length linkage on $k\geq 3$ vertices $q_1, \ldots, q_k$ together with a planar configuration $\configuration_{\mathcal{Q}}$ of $\mathcal{Q}$ (again, in the plane) such that no two adjacent edges are collinear. %
In the following, we identify every vertex $q_i$ with the point $\configuration_\mathcal{Q}(q_i)$, and define $R_i \coloneqq R(q_iq_{i+1}, \configuration_\mathcal{Q}, \varepsilon)$ for every $i \in [k -1]$.
For $1 < i < k$, let $H_i$ be the convex hull of the short sides of $R_{i-1}$ and $R_i$ incident to $q_i$.
For convenience, we define $H_1$ to be the rectangle with side length $\varepsilon$ obtained by extending $R_1$ by $\varepsilon/2$ beyond $q_1$.
We define $H_k$ analogously for~$R_{k-1}$ and $q_k$.
With this, we can define $S_i \coloneqq H_{i} \cup R_i \cup H_{i + 1}$ as the strip obtained from the $\varepsilon$-rectangle and the incident convex hulls, for every $i \in [k - 1]$.
Observe that every strip is a convex polygon.
The \emph{$\varepsilon$-thick corridor} $Q_{\varepsilon}$ of $\mathcal{Q}$ is then the union of all $S_i$, $i \in [k - 1]$.
Intuitively, we replace every edge of $\mathcal{Q}$ by its $\varepsilon$-rectangle and connect the endpoints of the rectangles by an edge; see also \Cref{fig:thick-corridor}a for an example.
Although we do not require $\configuration_{\mathcal{Q}}$ to use rational coordinates, we can still approximate the corridor with a polygon on rational coordinates.
\begin{observation}
	\label{obs:thick-corridors-rational}
	Let $\varepsilon > 0$ be rational and let $0 < \varepsilon' < \varepsilon$.
	There is an open polygon $P_{\varepsilon}$ with rational vertex coordinates such that $Q_{\varepsilon'} \subseteq P_{\varepsilon} \subseteq Q_{\varepsilon}$.
\end{observation}

\begin{figure}
    \centering
    \includegraphics[page=1]{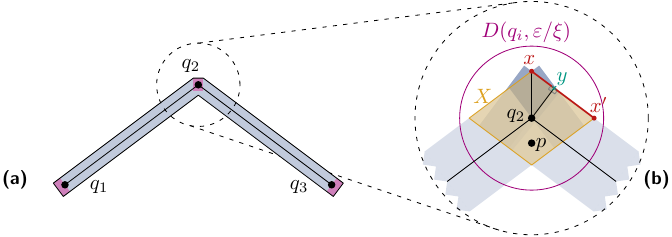}
    \caption{\textbf{\textsf{(a)}} An $\varepsilon$-thick corridor for the path on $q_1$, $q_2$, and $q_3$.
    Strips are colored blue and the polygons $H_i$ purple.
    \textbf{\textsf{(b)}}~Visualization of the proof of \Cref{lem:segment-intersection}. Extended strips are indicated in dark blue.}
    \label{fig:thick-corridor}
\end{figure}

We now want to show that, under mild assumptions, a unit-length linear linkage \linkage that traverses $Q_{\varepsilon}$ has to place its vertices ``near'' vertices of $\mathcal{Q}$.
For $1 < i < k$, let $0 < \alpha_i < \pi$ be the \emph{turning angle} of the two segments $q_{i - 1}q_{i}$ and $q_iq_{i + 1}$ under $\configuration_{\mathcal{Q}}$ at the vertex $q_i$; recall that we require that no two adjacent edges of $\mathcal{Q}$ are collinear under $\configuration_{\mathcal{Q}}$.
For a radius $r$ and point $p$, we let $D(p, r)$ denote the closed disk of radius~$r$ centered at $p$.
Recall that $H_j$, $1 < j < k$, is the convex hull of the short sides of $R_{j-1}$ and~$R_j$ incident to $q_j$.
For $i \in [k-1]$, we let $S_i^+$ be the rectangle of width $\varepsilon$ and length $1 + \varepsilon$ obtained by extending $R_i$ by $\varepsilon/2$ beyond each of its short sides.
Since the short sides have length $\varepsilon$ and their midpoint at $q_i$, we have $H_i \subseteq D(q_i, \varepsilon / 2)$ (for $1 < i < k$).
Furthermore, we have $S_i \subseteq S_i^+$ for every $i \in [k - 1]$.
Next, we show that two consecutive strips intersect inside a small disk; see also \Cref{fig:thick-corridor}b.
To this end, let $\xi = \min_{1 < i < k} \sin(\alpha_i)$ and observe $0 < \xi < 1$ as $0 < \alpha_i < \pi$.
\begin{lemma}
	\label{lem:segment-intersection}
	For every $1 < i < k$ it holds $S_{i-1} \cap S_i \subseteq D(q_i, \varepsilon/\xi)$.
\end{lemma}
\begin{proof}
	Consider a point $p \in S_{i-1} \cap S_i$.
	Since $S_{i-1} \subseteq S_{i-1}^+$ and $S_i \subseteq S_i^+$, we get $p \in S_{i-1}^+ \cap S_i^+$.
	Every point $x$ in $S_{i-1}^+$ has distance at most $\varepsilon / 2$ to the line through $q_{i - 1}$ and $q_i$.
	Likewise, every point $x'$ in $S_{i}^+$ has distance at most $\varepsilon / 2$ to the line through $q_{i}$ and $q_{i + 1}$.
	Since $p$ is a point in both strips, it has distance at most $\varepsilon / 2$ to both lines.
	Observe that both lines intersect at $q_i$ with angle $\alpha_i$.
	Thus, $p$ must lie in a rhombus~$X$ centered at $q_i$.
	Consider \Cref{fig:thick-corridor}b for the following arguments.
	First, note that the rhombus~$X$ has height~$\varepsilon$ as every side is $\varepsilon/2$ away from $q_i$.
	The points inside $X$ that are farthest away from~$q_i$ are the corners of~$X$.
	Let $xx'$ be a side of~$X$ and consider the endpoint of the rectangle~$R_j$ where~$S_j^+$ defines the side $xx'$, $j \in \{i - 1, i\}$.
	Furthermore, let $y$ be the point on $xx'$ perpendicular to~$q_i$.
	Observe that the distance of $y$ to $q_i$ is $\varepsilon / 2$ and we have $\angle xyq_i = \angle x'yq_i = \pi /2$.
	Moreover, assume without loss of generality that the inner angle of $X$ at $x$ is $\alpha_i$ and thus the angle at $x'$ is $\pi - \alpha_i$ (otherwise, we can swap the role of $x$ and $x'$ in the following arguments).
	Then, the distance $d_x$ between $x$ and $q_i$ is
	\begin{align*}
		\sin(\alpha_i / 2) = \frac{\varepsilon / 2}{d_x}\quad \Longleftrightarrow\quad d_x = \frac{\varepsilon}{2 \sin(\alpha_i / 2)}.
	\end{align*}
	By similar calculations, we obtain that the distance $d_{x'}$ between $x'$ and $q_i$ is $\frac{\varepsilon}{2 \sin((\pi - \alpha_i) / 2)} = \frac{\varepsilon}{2 \cos(\alpha_i / 2)}$.
	Using the fact that $\sin(2\beta) = 2\sin(\beta)\cos(\beta)$ holds, we derive
	\begin{align*}
		d_x = \frac{\varepsilon}{2 \sin(\alpha_i / 2)} = \frac{\varepsilon\cdot \cos(\alpha_i / 2)}{2 \sin(\alpha_i / 2)\cos(\alpha_i / 2)} = \frac{\varepsilon\cdot \cos(\alpha_i / 2)}{\sin(\alpha_i)} &\leq \frac{\varepsilon}{\sin(\alpha_i)}\\
		d_{x'} = \frac{\varepsilon}{2 \cos(\alpha_i / 2)} = \frac{\varepsilon \cdot \sin(\alpha_i/2)}{2 \cos(\alpha_i / 2)\sin(\alpha_i/2)} = \frac{\varepsilon \sin(\alpha_i / 2)}{\sin(\alpha_i)} &\leq \frac{\varepsilon}{\sin(\alpha_i)},
	\end{align*}
	where we used for the last inequality that $0 < \sin(\alpha_i / 2), \cos(\alpha_i / 2) < 1$ as $0 < \alpha_i < \pi$.
	Finally, observe that $\xi \leq \sin(\alpha_i)$, which implies that $x,x' \in D(q_i, \varepsilon / \xi)$ (and consequently $X \subseteq D(q_i, \varepsilon / \xi)$).
\end{proof}

We now use \Cref{lem:segment-intersection} to show that any unit-length segment inside $Q_\varepsilon$ must have its endpoints ``near'' $q_i$ and $q_{i + 1}$ for a suitable $i \in [k - 1]$.
To this end, recall $\xi = \min_{1 < i < k} \sin(\alpha_i) > 0$.
We define $r_{\mathcal{Q}} = 4\varepsilon / \xi$.

\begin{lemma}
	\label{lem:thick-corridors-unit-segment}
	Let $Q_{\varepsilon}$ be an $\varepsilon$-thick corridor on the vertices $q_1, \ldots, q_k$ such that $S_i \cap S_j = \emptyset$ for all $i, j \in [k-1]$ with $\lvert i - j\rvert \geq 2$, and $\varepsilon \leq \xi^2/16$.
	Let $\sigma \subseteq Q_{\varepsilon}$ be a segment of length one.
	Then there is an $i \in [k-1]$ such that one endpoint of $\sigma$ lies in $D(q_i, r_{\mathcal{Q}})$ and the other in $D(q_{i+1}, r_{\mathcal{Q}})$.
\end{lemma}
\begin{proof}
	For every $i \in [k - 1]$, let $\sigma_i \coloneqq \sigma \cap S_i$ be the part of the segment $\sigma$ that lives in the strip $S_i$.
	Observe that we have $\sigma = \bigcup_{i \in [k-1]} \sigma_i$ but the $\sigma_i$ do not form a partition of $\sigma$ (in particular, $\sigma_i\cap\sigma_{i + 1} \neq \emptyset$ can hold since $S_i$ and $S_{i + 1}$ are not disjoint).
	Since only neighboring strips are intersecting, we have that $\sigma = \sigma_a \cup \ldots \cup \sigma_b$ with $a, b\in [k-1]$.
	Moreover, since $\sigma$ has length one and does not leave $Q_\varepsilon$, we can observe that it cannot be in four strips, i.e., $b - a \leq 2$.
    To see this, assume that $\sigma$ would touch four strips.
    Since the intersection of two strips is contained in a disk of radius $\varepsilon / \xi$ by \Cref{lem:segment-intersection}, the segment $\sigma$ is for a length of at least $2 - 4(\varepsilon / \xi)$ contained in a single strip (first in $S_{a + 1}$ then in $S_{a + 2}$, and the strips $S_a$, $S_{a + 1}$, $S_{a + 2}$ and $S_{a + 3}$ are visited in this order).
    Since $\varepsilon \leq \xi^2/16$, we have $2 - 4(\varepsilon / \xi) \geq 2 - \xi / 4 > 1$, which is impossible.
    Furthermore, note that $b = a + 2$ can hold since at every $q_i$ two strips intersect.
	This gives rise to the following three cases, where we let $p$ and $p'$ be the endpoints of $\sigma$.
	\medskip
		
	\noindent\textit{Case 1: $b = a$.}
	From $b = a$, we can conclude $\sigma \subseteq S_a \subseteq S_a^+$.
	In particular, $\sigma$ lies inside a rectangle of width $\varepsilon$ and length $1 + \varepsilon$.
	Assume without loss of generality that $q_a = (0,0)$ and $q_{a+1} = (1,0)$.
	Let $\Delta_x$ and $\Delta_y$ be the difference in $x$- and $y$-coordinates of $p$ and $p'$, respectively.
	Since $\sigma \subseteq S_a^+$, we have $\Delta_y \leq \varepsilon$.
	Thus, $\Delta_x \geq \sqrt{1 - \varepsilon^2} \geq 1 - \varepsilon^2$.
	Assume without loss of generality $x(p) \leq x(p')$.
	As $\sigma \subseteq S_a^+$, we get $-\varepsilon/2 \leq x(p) \leq \varepsilon/2 + \varepsilon^2$, where the $\varepsilon/2$ comes from the fact that $S_a^+$ is $\varepsilon / 2$ longer than $S_a$ in both directions.
	Thus, $\Size{x(p) - x(q_a)} \leq \varepsilon/2 + \varepsilon^2$.
	Also $\Size{y(p) - y(q_a)} \leq \varepsilon/2$, which implies $p \in D(q_{a}, r_{\mathcal{Q}})$ as the Manhattan distance between $p$ and $q_a$ is at most $\varepsilon + \varepsilon^2$.
	Symmetrically, we can show $p' \in D(q_{a + 1}, r_{\mathcal{Q}})$, which yields the statement in this case.
	\medskip
	
	\noindent\textit{Case 2: $b = a + 1$.}
	Now $\sigma$ touches the two strips $S_a$ and $S_{a + 1}$.
	If $\sigma$ lies entirely in one of them, we can use the same arguments as in Case 1.
	Thus, we assume that one endpoint of $\sigma$, say $p$, lies only in $S_a$ and the other endpoint, i.e., $p'$, lies only in $S_{a + 1}$.
	This implies that there exists a point $c \in \sigma$ on the segment that lies in $S_a \cap S_{a + 1}$.
	As $\sigma$ is a straight segment (and $S_a$ and $S_{a + 1}$ are convex), we can pick a $c$ such that~$\sigma$ splits into the two segments $pc$ and $cp'$, where $pc \subseteq S_a$ and $cp' \subseteq S_{a + 1}$ holds.
	Observe that we have $c \in D(q_{a+1}, \varepsilon/\xi)$ by \Cref{lem:segment-intersection}.

	We now show that $\Vert p' - c \Vert_2 \leq 2\varepsilon/\xi$ or $\Vert c - p \Vert_2 \leq 2\varepsilon/\xi$ holds, i.e., one of the two sub-segments is ``short''.
	We assume that $\Vert p' - c \Vert_2 \leq \Vert c - p \Vert_2$, i.e., $cp'$ is the shorter segment; in the other case we can use symmetric arguments.
	Note that this implies $\Vert c - p \Vert_2 \geq 1/2$ and see \Cref{fig:thick-corridor-strip} for the following arguments.
    \begin{figure}
        \centering
        \includegraphics[page=2]{thick-corridor}
        \caption{
        Visualization of Case~2 in the proof of \Cref{lem:thick-corridors-unit-segment}. The orange polygon is the intersection of the two strips.}
        \label{fig:thick-corridor-strip}
    \end{figure}
	Consider a line $l$ that runs parallel to $q_aq_{a + 1}$ and intersects $p$.
	Let $z$ be the intersection point of $l$ and a line that is orthogonal to $l$ and intersects $c$.
	The points $p$, $c$, and $z$ form a right-angle triangle.
	Let $\beta$ be the angle at the vertex $p$, which is the angle spanned by the segments $pz$ and $pc$.
	Observe that $pc \subseteq S_a \subseteq S_a^+$, which implies that the segment $cz$ has length at most $\varepsilon$.
	Hence, $\Vert c - p\Vert_2 \sin(\beta) \leq \varepsilon$.
	Together with $\Vert c - p \Vert_2 \geq 1/2$, this implies $\sin(\beta) \leq 2\varepsilon$.
	Now consider the acute angle $\beta'$ between (the lines through)~$\sigma$ and $q_{a+1}q_{a+2}$.
	Using symmetric arguments, we can conclude $\Vert p' - c\Vert_2 \sin(\beta') \leq \varepsilon$.
	Recall that the segments $q_aq_{a+1}$ and $q_{a+1}q_{a+2}$ meet at the angle $\alpha_{a + 1}$ and we have $\sin(\alpha_{a + 1}) \geq \xi$.
    Hence, the (infinite) lines through these segments form an acute angle $\alpha_{a + 1}'$ with $\alpha_{a + 1}' = \min \{\alpha_{a + 1}, \pi - \alpha_{a + 1}\} \leq \pi / 2$.
    Observe $\sin(\alpha_{a + 1}') = \sin(\alpha_{a + 1})$.
    Since $\beta$ and $\beta'$ are/can be defined via lines through $\sigma$, $q_aq_{a+1}$, and $q_{a+1}q_{a+2}$, $\beta + \beta' \geq \alpha_{a + 1}'$ holds.
	From $\sin(\beta) \leq 2\varepsilon$, and $\varepsilon \leq \xi^2 / 16$, we get $2\varepsilon \leq \xi^2 / 8$ and thus $\sin(\beta) \leq 2\varepsilon \leq \xi/2 \leq \sin(\alpha_{a + 1}') / 2$ (since $\sin(\alpha_{a + 1}') = \sin(\alpha_{a + 1})$).
	Since $\beta + \beta' \geq \alpha_{a + 1}'$, we get $\sin(\beta') \geq \sin(\alpha_{a + 1}')  / 2 \geq \xi/2$.
	Plugging this into $\Vert p' - c\Vert_2 \sin(\beta') \leq \varepsilon$ (which we have derived before) yields $\Vert p' - c\Vert_2 \leq 2\varepsilon / \xi$.

	We now use this to show $p \in D(q_{a}, r_\mathcal{Q})$ and $p' \in D(q_{a+1}, r_\mathcal{Q})$.
	Using the triangle inequality, we get $\Vert p' - q_{a+1}\Vert_2 \leq \Vert p' - c\Vert_2 + \Vert c - q_{a+1}\Vert_2 \leq 3\varepsilon/\xi \leq r_{\mathcal{Q}}$; recall $c \in D(q_{a+1}, \varepsilon/\xi)$ by \Cref{lem:segment-intersection}.
	Hence, $p' \in D(q_{a+1}, r_\mathcal{Q})$.
	For the other endpoint, i.e., $p$, we again use the triangle inequality $\Vert p' - p \Vert_2 \leq \Vert p - q_{a + 1} \Vert_2 + \Vert p' - q_{a + 1}\Vert_2$ which can be reformulated to $\Vert p' - p \Vert_2 - \Vert p' - q_{a + 1}\Vert_2 \leq \Vert p - q_{a + 1} \Vert_2$.
	This is, by the above derivation, at least $1 - 3\varepsilon/\xi$.
	We can again assume without loss of generality that $q_a = (0,0)$ and $q_{a+1} = (1,0)$.
	Thus, $\Size{y(p)} \leq \varepsilon / 2$ (since $p$ lies in the strip $S_a^+$) which, due to similar calculations as in Case 1, implies $\Size{x(p) - x(q_{a + 1})} \geq 1 - 3\varepsilon/\xi - \varepsilon/2$ (since $-\varepsilon / 2 \leq x(p) \leq 1 + \varepsilon / 2$).
	This implies $p \in D(q_{a}, r_{\mathcal{Q}})$ as the Manhattan distance between $p$ and $q_a$ is at most $3\varepsilon/\xi + \varepsilon \leq r_{\mathcal{Q}}$ (recall $0 < \xi < 1$).
	Hence, $p \in D(q_{a}, r_\mathcal{Q})$, which shows the statement.
	\medskip
	
	\noindent\textit{Case 3: $b = a + 2$.}
	Observe that the last case implies the existence of points $c \in S_a \cap S_{a + 1}$ and $c' \in S_{a + 1} \cap S_{a + 2}$ such that $\sigma$ traverses the points $p$, $c$, $c'$, and $p'$ in this order, i.e., $\sigma$ can be split into $pc$, $cc'$, and $c'p'$.
	Intuitively, this means that $\sigma$ traverses the entire strip $S_{a+1}$.
	More formally, we can use the triangle inequality $1 = \Vert q_{a+1} - q_{a+2}\Vert_2 \leq \Vert q_{a+1} - c\Vert_2 + \Vert c' - c\Vert_2 + \Vert c' - q_{a+2}\Vert_2$ to conclude $1 \leq \Vert c' - c\Vert_2 + 2\varepsilon / \xi$; recall $c \in D(q_{a+1}, \varepsilon/\xi)$ and $c' \in D(q_{a+2}, \varepsilon/\xi)$ by \Cref{lem:segment-intersection}.
	As also $\sigma$ has length one, the segments $pc$ and $c'p'$ together can have length at most $2\varepsilon / \xi$.
	Using the same triangle inequality as for Case 2, we can use this to derive $p \in D(q_{a + 1}, r_\mathcal{Q})$ and $p' \in D(q_{a + 2}, r_\mathcal{Q})$, which yields the claim.%
\end{proof}

We now use \Cref{lem:thick-corridors-unit-segment} to show that a configuration inside an $\varepsilon$-thick corridor must traverse it by placing each $v_i$ ``near'' $q_i$.
To this end, we define $A_i \coloneqq D(q_i, r_{\mathcal{Q}})$ for $i \in [k]$.
These \emph{areas} $A_i$ formalize the above-hinted concept of ``near''.
In the following lemma, we show that, under mild assumptions, a linkage which places its first vertex in $A_1$ must place its last vertex in $A_k$.

\begin{lemma}
	\label{lem:thick-corridors-correctness}
	Let $Q_{\varepsilon}$ be an $\varepsilon$-thick corridor on the vertices $q_1, \ldots, q_k$, $k \geq 3$ such that
	the strips $S_i$ and $S_j$ are disjoint for all $i,j \in [k-1]$ with $\Size{j - i} \geq 2$,
	every area $A_i$ only intersects $S_{i - 1}$ and $S_i$ (if these strips exist) and no other area,
	and $\varepsilon \leq \xi^2 / 16$.	
	Moreover, let $\linkage$ be a unit-length linear linkage on $k$ vertices $v_1, \ldots, v_k$ with a configuration $\configuration$ that lives inside $Q_{\varepsilon}$.
	If $\configuration(v_1) \in A_1$ and for every $i \in [k]$ there is at most one $j \in [k]$ such that $\configuration(v_j) \in A_i$, then $\configuration(v_i) \in A_i$ for every $i \in [k]$.
\end{lemma}
\begin{proof}
	We show the statement by induction on $i$.
	The base case $i = 1$ holds by assumption.
	Thus, we assume that the statement holds for every $i' \leq i < k$ and show that this implies $\configuration(v_{i +1}) \in A_{i + 1}$.
	Let $\sigma$ be the segment $\configuration(v_iv_{i+1})$.
	By assumption, we have that $\sigma \subseteq Q_{\varepsilon}$.
	Using \Cref{lem:thick-corridors-unit-segment}, we know that there exists a $j \in [k - 1]$ such that one endpoint of $\sigma$ lies in $D(q_j, r_{\mathcal{Q}}) = A_j$ and the other one in $D(q_{j+1}, r_{\mathcal{Q}}) = A_{j+1}$.
	By our induction hypothesis, we know that $\configuration(v_i) \in A_i$ holds.
	Moreover, the areas are pairwise disjoint which implies that $\configuration(v_{i + 1}) \in A_{i - 1} \cup A_{i + 1}$ (with $A_0 = \emptyset$).
	Since there is at most one vertex in every area, and by our induction hypothesis we have $\configuration(v_{i - 1}) \in A_{i - 1}$, we conclude $\configuration(v_{i + 1}) \in A_{i + 1}$, which completes the proof.
\end{proof}

Finally, we observe that not only the configuration $\configuration$ with $\configuration(v_i) = q_i$ can pass through $Q_{\varepsilon}$:

\begin{observation}
	\label{obs:thick-corridors-witness}
	Let $Q_{\varepsilon}$ be an $\varepsilon$-thick corridor on the vertices $q_1, \ldots, q_k$, $k \geq 3$.
	There is a $\tau > 0$ such that the following holds for every $p \in \mathbb{R}^2$ with $\Vert p\Vert_2 \leq \tau$.
	For a unit-length linear linkage \linkage on $k$ vertices $v_1, \ldots, v_k$, the drawing \configuration with $\configuration(v_i) = q_i + p$ for every $i \in [k]$ is a planar configuration of \linkage that lives inside $Q_{\varepsilon}$.
\end{observation}
Note that we can formulate \Cref{obs:thick-corridors-witness} also with respect to the refined $\varepsilon'$-thick corridor $Q_{\varepsilon'}$ from \Cref{obs:thick-corridors-rational}.

While the $\varepsilon$-corridor $Q_{\varepsilon}$ is defined using a planar configuration $\configuration_{\mathcal{Q}}$, in our gadgets two \emph{different} corridors might intersect.
The next lemma shows that for sufficiently small $\varepsilon$, a linkage that passes through one corridor cannot switch into the other one.

\begin{lemma}
	\label{lem:crossing-strips}
	Let $Q_{\varepsilon}$ be an $\varepsilon$-thick corridor on the vertices $q_1, \ldots, q_k$ satisfying the conditions of \Cref{lem:thick-corridors-correctness}, and let $S$ be the union of strips from other $\varepsilon$-thick corridors.
	Let $0 < \mu \le 1$ be such that, for every $j \in [k]$, the disk $D(q_j, \mu)$ is disjoint from $S$ and from every strip of $Q_\varepsilon$ other than $S_{j-1}$ and $S_j$, and let $\varepsilon \leq \mu^2\xi/18$.
	If $\sigma \subseteq Q_\varepsilon \cup S$ is a segment of length one with an endpoint in an area $A_i$, $i \in [k]$, then $\sigma \subseteq Q_\varepsilon$.
\end{lemma}
\begin{proof}
	Let $p \in A_i$ and $q$ be the endpoints of $\sigma$ and assume that $\sigma$ contains a point $x \notin Q_{\varepsilon}$.
	Hence, $x \in S\setminus Q_{\varepsilon}$ and $\Vert x - q_j\Vert_2 > \mu$ for every $j \in [k]$.
	
	First, we show that $\sigma$ ``starts'' inside a single strip.
	Let $\sigma' \subseteq \sigma$ consist of all points $c$ on $\sigma$ with $r_{\mathcal{Q}} + \varepsilon/\xi < \Vert c - p\Vert_2 \leq \mu - r_{\mathcal{Q}}$.
	It has length $\mu - 2r_{\mathcal{Q}} - \varepsilon/\xi = \mu - 9\varepsilon/\xi \geq \mu/2$ (recall $\varepsilon \leq \mu^2\xi/18$).
	Since $\Vert p - q_i \Vert_2 \leq r_{\mathcal{Q}}$, every point of $\sigma'$ has distance more than $\varepsilon/\xi$ and at most $\mu$ from $q_i$.
	By the choice of $\mu$, the disk $D(q_i, \mu)$ is disjoint from $S$ and meets no strip other than $S_{i-1}$ and $S_i$, so $\sigma' \subseteq S_{i-1} \cup S_i$.
	As these two strips only overlap in $D(q_i, \varepsilon/\xi)$ by \Cref{lem:segment-intersection} and $\sigma'$ is connected, $\sigma'$ lies entirely in one of them.
	Without loss of generality, we assume $\sigma' \subseteq S_i$; the other case is symmetric.
	Let $a$ and $b$ be the endpoints of $\sigma'$, where $b$ is the one that is farther away from $p$ and observe that $\sigma$ visits the points $p$, $a$, $b$, $x$, $q$ in this order.
	
	Next, we show that $\sigma$ is ``almost parallel'' to $q_iq_{i+1}$.
	Let $l$ be the infinite line through $q_iq_{i+1}$.
	Along the segment $\sigma$, the (signed) distance to $l$ changes continuously.
	Since $a, b \in S_i$ both lie within distance $\varepsilon/2$ of~$l$, it changes by at most $\varepsilon$ between $a$ and $b$, i.e., over a length of at least $\mu/2$.
	Hence, it changes by at most $2\varepsilon/\mu$ along all of $\sigma$, which has length one.
	In other words, $\sigma$ is tilted against~$l$ by an angle $\beta \leq \pi/2$ with $\sin(\beta) \leq 2\varepsilon/\mu$, and, as $b$ lies within $\varepsilon/2$ of~$l$, every point of $\sigma$ lies within distance $d = \varepsilon/2 + 2\varepsilon/\mu$ of~$l$.
    We now argue that $\sigma$ also goes towards $q_{i+1}$.
    To this end, observe that the strip $S_i$ extends only $\varepsilon/2$ to either side of $l$, and its points behind $q_i$ lie in $H_i \subseteq D(q_i, \varepsilon)$ (for $1 < i < k$ we have already observed that $H_i \subseteq D(q_i, \varepsilon / 2)$ and $H_1$ and $H_k$ are rectangles of side length $\varepsilon$ and $\varepsilon/2$).
	Since $b \in S_i$ has distance at least $\Vert b - p\Vert_2 - \Vert p - q_i\Vert_2 \geq \mu - 2r_{\mathcal{Q}} \geq \varepsilon$ from $q_i$ (as $r_{\mathcal{Q}} \leq 4\mu / 18$ and $\varepsilon \leq \mu^2\xi/18$), it therefore lies at least $\mu - 2r_{\mathcal{Q}} - \varepsilon/2 > r_{\mathcal{Q}}$ in front of~$q_i$ (in direction of $q_{i+1}$), whereas $p \in D(q_i, r_{\mathcal{Q}})$ lies at most $r_{\mathcal{Q}}$ in front of~$q_i$.
	As $b$ lies between~$p$ and $q$ on $\sigma$, the segment $\sigma$ runs towards $q_{i+1}$.
    Hence, we can obtain~$\sigma$ by placing a copy of the edge $q_iq_{i+1}$ at $p$ and rotating it (around $p$) by the angle~$\beta$.
    Thus, the endpoint~$q$ of $\sigma$ can be described as $p + u$, where $u$ is the unit-vector obtained by rotating the vector $q_{i + 1} - q_i$ by~$\beta$.
    Therefore, we can bound the distance between $q_{i + 1} - q_i$ and $u$ by $\beta \leq (\pi/2)\sin(\beta)$.
    Moreover, the distance between~$p$ and $q_i$ is at most $r_{\mathcal{Q}}$.
    Using $q - q_{i + 1} = (p - q_i) + (u - (q_{i + 1} - q_i))$ together with the triangle inequality,
    $(\pi/2)\sin(\beta)\leq \pi\varepsilon/\mu \leq \pi\mu/18$, and $r_{\mathcal{Q}} \leq 4\mu/18$, we get $\Vert q - q_{i+1}\Vert_2 \leq r_{\mathcal{Q}} + \pi\varepsilon/\mu \leq \mu/2$.
	As $S$ is disjoint from $D(q_{i+1}, \mu)$, we get $q \notin S$.
	
	Finally, $\sigma$ has to get from $x$, which has distance more than $\mu$ from $q_{i+1}$, to $q$, which has distance less than $\mu/2$ from $q_{i+1}$.
	Hence, it passes a point $z$ between $x$ and $q$ with $\Vert z - q_{i+1}\Vert_2 = \mu/2$.
	By the choice of $\mu$, we have $z \notin S$ and thus $z \in S_i \cup S_{i+1}$.
	However, $z \notin S_{i+1}$:
	recall that every point of $\sigma$ lies within distance $d = \varepsilon/2 + 2\varepsilon/\mu$ of $l$.
	We can show, analogous to \Cref{lem:segment-intersection}, that every point within $d$ of both $l$ and the line through $q_{i+1}q_{i+2}$ lies in $D(q_{i+1}, 2d/\xi)$, where $2d/\xi = (\varepsilon + 4\varepsilon/\mu)/\xi \leq 5\varepsilon/(\mu\xi) \leq 5\mu/18 < \mu/2$.
	Thus $z \in S_i$, i.e., $\sigma$ would have to re-enter $S_i$ after leaving it at $x$.
	As $b, z \in S_i$, $x$ lies between them, and~$S_i$ is convex, we get $x \in S_i \subseteq Q_\varepsilon$, a contradiction.
\end{proof}

\paragraph*{Entrance and Exit Areas.}
The tunnel, clause, and variable gadgets that we will construct in the upcoming \Cref{sec:hardness-details} are composed of several $\varepsilon$-thick corridors of some length $k$ for a suitable value of~$\varepsilon$.
By definition, each of them has exactly $k$ areas $A_i$.
Two of them are important for us, which are the \emph{entrance} and \emph{exit} areas.
For a gadget $F$ related to a variable $x_i$, we denote them as $\PositiveEntrance{i}(F)$/$\NegativeEntrance{i}(F)$ and $\PositiveExit{i}(F)$/$\NegativeExit{i}(F)$, respectively. 
Intuitively speaking, these are the areas where the linkage enters a gadget~$F$ and, thanks to \Cref{lem:thick-corridors-correctness}, exits $F$ again.
If an area (or configuration) is related to $x_i$ being true or false, we color the area (configuration) in green (\ExampleArea[PositiveColor][1]) or red (\ExampleArea[NegativeColor][0.5]) in the figures, respectively.
In the end, we create the polygonal domain $P$ by gluing the individual gadgets together at the correct areas.
Furthermore, we will argue that the preconditions for \Cref{lem:thick-corridors-correctness} are fulfilled in our reduction.
In the next section, we describe each of the gadgets in detail.
Note that all gadgets are agnostic to translations and rotations by~$90^{\circ}$.
Furthermore, we use one global value for $\varepsilon$, $\xi$, and $r$ (which takes the role of $r_{\mathcal{Q}}$), whose value we specify in \Cref{sec:hardness-combine}.

\subsection{The Gadgets of the Hardness Reduction}
\label{sec:hardness-details}
In this section, we describe the individual gadgets that we employ in our hardness reduction to encode variables, clauses, and transport truth assignments through the plane; recall that we reduce from \probname{Planar Monotone 3-Sat}.
In the end, the polygonal domain $P$ will consist of the interior of all $\varepsilon$-thick corridors that we define to construct the following gadgets.
Let $F$ be one of these gadgets.
We let $m_F$ denote the number of edges for a route through $F$ (in most gadgets, every route through $F$ will have the same number of edges).
Moreover, we say that a configuration $\configuration$ of a unit-length linear linkage $\linkage$ on $n$ vertices $v_1, \ldots, v_n$, that lives inside $P$ \emph{passes through} $F$ starting at a vertex $v_i$ if $\configuration(v_i)$ lies in an entrance area of $F$, $i + m_F \leq n$, $\configuration(v_{i+1})$ lies in an area of $F$, and for every area $A$ of $F$ there is at most one $j$ with $i \leq j \leq i + m_F$ such that $\configuration(v_{j}) \in A$, where $m_F$ is the length of the route from the respective entrance area to the corresponding exit area.
Intuitively, a configuration passes through a gadget if it does not ``waste'' vertices of \linkage inside $F$ by, e.g., leaving the gadget at the entrance through which it entered.
Furthermore, observe that this resembles one of the preconditions of \Cref{lem:thick-corridors-correctness}.
In the proof of \Cref{thm:hardness}, we show that every possible solution \configuration can only pass through the gadgets.

\paragraph{Edge Gadget.}
The first gadget that we describe is the edge gadget, which comes in three different types: tunnels, bends, and shifters.

\begin{figure}
    \centering
    \includegraphics[page=4]{edge-gadget}
    \caption{\textbf{\textsf{(a)}} A tunnel,~\textbf{\textsf{(b)}} a bend,~and \textbf{\textsf{(c)}} a shifter with configurations through them. We hatch in this and the following figures the outside of the polygonal domain if there is risk of confusion. Moreover, we indicate ``obstacles'' inside the polygonal domain in gray and the entrance and exit areas (not to scale) with dashed circles.}
    \label{fig:edge-gadget}
\end{figure}

\medskip

\noindent\textit{Tunnel.}
We discuss tunnels that run horizontally from left to right; other axis-aligned tunnels can be constructed by rotating or reflecting the following construction.
A tunnel $T$ is a sawtooth-like shaped part of the polygonal domain $P$ of height $0.6$ and width $1.6$.
It is composed of two $\varepsilon$-thick corridors, the first one on the path $((0,0), (0.8, 0.6), (1.6,0))$ and the second one on the path $((0,0.6), (0.8, 0), (1.6,0.6))$; see also \Cref{fig:edge-gadget}a.
We can create a long tunnel by piecing smaller ones together to obtain concatenated paths.
Observe that any configuration \configuration that passes through a tunnel must \emph{zig-zag} inside it by alternating the placement of the vertices between a placement near the top and the bottom of the tunnel.
This allows us to define two equivalence classes on the configurations depending on the side of the tunnel where they place the first vertex.
They correspond to the truth assignment to a variable $x_i \in \mathcal{X}$.
Assuming that the tunnel is later related to the variable $x_i \in \mathcal{X}$, we create the following areas; recall also \Cref{fig:edge-gadget}a.
\begin{align*}
	\NegativeEntrance{i}(T) &= D((0,0), r)
	&\quad\quad
	\NegativeExit{i}(T) &= D((1.6,0), r)\\ 
	\PositiveEntrance{i}(T) &= D((0,0.6), r)
	&\quad\quad
	\PositiveExit{i}(T) &= D((1.6,0.6), r)
\end{align*}

We now show the crucial property of the tunnels.
\begin{lemma}
	\label{lem:edge-gadget-tunnel}
	Let $T$ be a tunnel related to a variable $x_i \in \mathcal{X}$. 
	Each of its two routes has $m_T = 2$ edges.
	Furthermore, let $\configuration$ be a configuration of a unit-length linear linkage on $n$ vertices $v_1, \ldots, v_n$ that passes through $T$ starting at $v_j$, $j \leq n - 2$.
	If $\configuration(v_j) \in \NegativeEntrance{i}(T)$
	(respectively $\configuration(v_j) \in \PositiveEntrance{i}(T)$), then
	$\configuration(v_{j+2}) \in \NegativeExit{i}(T)$
	(respectively $\configuration(v_{j + 2}) \in \PositiveExit{i}(T)$).
\end{lemma}
\begin{proof}
	That $m_T = 2$ holds can be confirmed with \Cref{fig:edge-gadget}a.
	We can observe that both routes satisfy the preconditions for \Cref{lem:thick-corridors-correctness,lem:crossing-strips} for appropriate values of $\varepsilon$ and $\xi$.
	In particular, the routes cross each other twice, near the midpoints of their strips.
	Thus, we can iteratively apply \Cref{lem:crossing-strips,lem:thick-corridors-correctness,lem:thick-corridors-unit-segment} to eventually conclude that $\configuration(v_{j+2})$ lies in the respective exit area.
\end{proof}

\medskip

\noindent\textit{Bend.}
Tunnels are accompanied by bends, which force the configuration \configuration to perform a $90^\circ$ turn and are depicted in \Cref{fig:edge-gadget}b.
All bends that we need are $90^\circ$ bends, which can be obtained by rotating/reflecting the following construction.
The main difficulty of a bend, in contrast to a tunnel, is to ensure that any configuration uses the same number of edges to perform the bend, no matter in which equivalence class it is.
In particular, in a na\"ive construction, one equivalence class places one vertex fewer inside the gadget.
The following construction does not have this issue, at the cost of requiring a larger area to perform the turn.
A bend consists of two parts.
The first part is composed of two $\varepsilon$-thick corridors.
The first corridor is composed of the path $((0,0), (0.8, 0.6), (1.6,0), (2.2, 0.8), (1.6, 1.6), (2.2, 2.4))$ and the second corridor is composed of the path $((0,0.6), (0.8, 0), (1.8,0), (2.4, 0.8), (1.6, 1.4), (1.6, 2.4))$; see also \Cref{fig:edge-gadget}b.
Note that the corridors force the green configuration to draw one edge (almost) horizontal and one (almost) vertical.
However, no two consecutive segments are collinear.
Furthermore, observe that the bend starts with half a tunnel.
The second part of the bend consists of a tunnel that we attach to ensure correctness of the gadget.

Observe that the bend has at its start and end a height (or width) of $0.6$, allowing us to attach tunnels on either of its ends.
Similar to a tunnel, we can define the following areas for a bend $B$ (with respect to the attached tunnel).
\begin{align*}
	\NegativeEntrance{i}(B) &= D((0,0), r)
	&\quad\quad
	\NegativeExit{i}(B) &= D((2.2,4), r)\\
	\PositiveEntrance{i}(B) &= D((0,0.6), r)
	&\quad\quad
	\PositiveExit{i}(B) &= D((1.6,4), r)
\end{align*}%
We now state the main properties of a bend.
\begin{lemma}
	\label{lem:edge-gadget-bend}    
	Let $B$ be a bend related to a variable $x_i \in \mathcal{X}$. 
	Each of its two routes has $m_B = 7$ edges.
	Furthermore, let $\configuration$ be a configuration of a unit-length linear linkage on $n$ vertices $v_1, \ldots, v_n$ that passes through $B$ starting at $v_j$, $j \leq n - 7$.
	If $\configuration(v_j) \in \NegativeEntrance{i}(B)$
	(respectively $\configuration(v_j) \in \PositiveEntrance{i}(B)$), then
	$\configuration(v_{j+7}) \in \NegativeExit{i}(B)$
	(respectively $\configuration(v_{j + 7}) \in \PositiveExit{i}(B)$).
\end{lemma}
\begin{proof}
	Similar to the proof of \Cref{lem:edge-gadget-tunnel}, we can observe that $m_B = 7$ holds by looking at \Cref{fig:edge-gadget}b (where we have five edges in the first part, depicted in \Cref{fig:edge-gadget}b, and two edges in the second part).
	Moreover, both routes satisfy the preconditions for \Cref{lem:thick-corridors-correctness,lem:crossing-strips} for appropriate values of $\varepsilon$ and $\xi$ in most places.
	
	However, in contrast to tunnels, bends have some endpoint used to define a corridor inside another corridor; recall \Cref{fig:edge-gadget}b.
	Observe that there the distance between the defining points of the corridors is~$0.2$.
	Hence, we can choose our $\varepsilon$ small enough such that the areas do not intersect.
	Then, the statement follows with similar arguments as in the proof of \Cref{lem:edge-gadget-tunnel}, i.e., by combining \Cref{lem:segment-intersection,lem:crossing-strips}, together with the following three observations.
	First, at $(1.6,0)$, the red configuration cannot enter the corridor for the green configuration as the adjacent turning angles are too large to facilitate a further straight edge.
	Similarly, the green configuration cannot bend enough to enter (and continue) in the corridor for the red configuration.
	Second, we can choose $\varepsilon$ small enough such that the two corridors do not intersect at $(2.2, 0.8)$ and $(2.4, 0.8)$ as in \Cref{fig:edge-gadget}b (but of course, they will intersect on their way towards $(1.6, 1.6)$).
	Third, at $(1.6,1.6)$, the red configuration (again) cannot enter the corridor for the green configuration since we can choose $\varepsilon$ small enough such that the tunnel which forms the second part of the bend does not permit the placement of another vertex in said corridor and the same applies to the green configuration.%
\end{proof}

\medskip 

\noindent\textit{Shifter.}
Tunnels and bends can only start and end at specific coordinates due to their construction.
With a shifter, we can shift tunnels up and down by $0.2$ to give us more flexibility in later gadgets.
Observe in \Cref{fig:edge-gadget}a that inside a tunnel the distance of the endpoints of an edge of the linkage is approximately~$0.8$ and~$0.6$ in $x$- and $y$-direction, respectively.
With the gadget from \Cref{fig:edge-gadget}c, we can force an inverted behavior to move a configuration over the course of two edges up (or down) by $0.2$.
In the following, we describe the upwards direction; the downwards direction can be created by mirroring this construction.
As before, a shifter consists of two $\varepsilon$-thick corridors.
The first is defined by the path $((0,0), (0.6,0.8), (1.4, 0.2))$ and the second by the path $((0,0.6), (0.8,0), (1.4, 0.8))$; see also \Cref{fig:edge-gadget}c.
For a shifter $S$ related to a variable $x_i \in \mathcal{X}$, the entrance and exit areas are defined as follows.
\begin{align*}
	\NegativeEntrance{i}(S) &= D((0,0), r)
	&\quad\quad
	\NegativeExit{i}(S) &= D((1.4, 0.2), r)\\
	\PositiveEntrance{i}(S) &= D((0,0.6), r)
	&\quad\quad
	\PositiveExit{i}(S) &= D((1.4,0.8), r)
\end{align*}

The following lemma shows that the shifter indeed fulfills its purpose.
\begin{lemma}
	\label{lem:edge-gadget-shifter}    
	Let $S$ be a shifter related to a variable $x_i \in \mathcal{X}$. 
	Each of its two routes has $m_S = 2$ edges.
	Furthermore, let $\configuration$ be a configuration of a unit-length linear linkage on $n$ vertices $v_1, \ldots, v_n$ that passes through $S$ starting at $v_j$, $j \leq n - 2$.
	If $\configuration(v_j) \in \NegativeEntrance{i}(S)$
	(respectively $\configuration(v_j) \in \PositiveEntrance{i}(S)$), then
	$\configuration(v_{j+2}) \in \NegativeExit{i}(S)$
	(respectively $\configuration(v_{j + 2}) \in \PositiveExit{i}(S)$).
\end{lemma}
\begin{proof}
	The statements can be shown using analogous arguments as for the tunnel; recall \Cref{lem:edge-gadget-tunnel} and see \Cref{fig:edge-gadget}c for the edge count of the two routes.
\end{proof}

\paragraph{Clause Gadget.}
The next gadget that we discuss is the gadget for a \emph{positive} clause $x_i \lor x_j \lor x_k$ (the gadget for a negative clause is described at the end, after \Cref{lem:clause-gadget}).
Recall that we can assume without loss of generality that every clause consists of three (not necessarily distinct) literals.
The clause gadget consists of several parts, where the main part $C'$ is depicted in \Cref{fig:clause-gadget}a.
It consists of seven (possibly intersecting) $\varepsilon$-thick corridors inside the gadget to limit how a configuration~\configuration can interact with it.
The first corridor $Q_{\varepsilon}^1$ is defined by the path $((0, 1.6), (0.6, 0.8), (0,0))$.
The second corridor~$Q_{\varepsilon}^2$ is defined by the path $((0, 1), z, (0,0.6))$, where $z$ is the intersection point between the two unit circles centered at $(0,1)$ and $(0,0.6)$ that lies inside the gadget, i.e., $z = (\sqrt{1 - 0.2^2}, 0.8)$.
The third and fourth corridors $Q_{\varepsilon}^3$ and $Q_{\varepsilon}^4$ are defined by reflecting the first and second one at $x = 1.4$.
Observe that these corridors can be traversed by a linear linkage of length two.
The remaining three corridors require a linear linkage of length three.
To this end, the fifth corridor $Q_{\varepsilon}^5$ is defined by the path $((0.2, 0), (0.9, \sqrt{1-0.7^2}), (1.9, \sqrt{1-0.7^2}), (2.6, 0))$, i.e., a path that spans an $x$-distance of $2.4$ by first going up diagonally, then going horizontally, and then going down to $(2.6, 0)$.
The sixth corridor $Q_{\varepsilon}^6$ is defined by the path $((0.8, 0), (0.7, \sqrt{1-0.1^2}), z', (2, 0))$, where $z'$ is the intersection point between the two unit circles centered at $(0.7, \sqrt{1-0.1^2})$ and $(2, 0)$ with larger $y$-coordinate.
The seventh corridor~$Q_{\varepsilon}^7$ can be obtained by mirroring the sixth one.
Observe that, for $\varepsilon$ small enough, $Q_{\varepsilon}^6$ intersects $Q_{\varepsilon}^2$ but not~$Q_{\varepsilon}^1$.
Symmetrically, $Q_{\varepsilon}^7$ intersects $Q_{\varepsilon}^4$ but not $Q_{\varepsilon}^3$.
Furthermore, $Q_{\varepsilon}^5$ does not intersect any corridor $Q_{\varepsilon}^i$ with $i = 1,2,3,4$.

In our reduction, the corridors $Q_{\varepsilon}^1$ and $Q_{\varepsilon}^2$ will represent $x_i$, $Q_{\varepsilon}^5$, $Q_{\varepsilon}^6$, and $Q_{\varepsilon}^7$ will represent $x_j$ and $Q_{\varepsilon}^3$ and $Q_{\varepsilon}^4$ will represent $x_k$.
This gives rise to the following entrance and exit areas of the main part $C'$.
\begin{align*}
	\NegativeEntrance{i}(C') &= D((0,1), r)
	&\quad\quad
	\NegativeExit{i}(C') &= D((0,0.6), r)\\
	\PositiveEntrance{i}(C') &= D((0,1.6), r)
	&\quad\quad
	\PositiveExit{i}(C') &= D((0,0), r)\\
	\NegativeEntrance{j}(C') &= D((0.8,0), r)
	&\quad\quad
	\NegativeExit{j}(C') &= D((2,0), r)\\
	\PositiveEntrance{j}(C') &= D((0.2,0), r)
	&\quad\quad
	\PositiveExit{j}(C') &= D((2.6,0), r)\\
	\NegativeEntrance{k}(C') &= D((2.8,0.6), r)
	&\quad\quad
	\NegativeExit{k}(C') &= D((2.8,1), r)\\
	\PositiveEntrance{k}(C') &=  D((2.8,0), r)
	&\quad\quad
	\PositiveExit{k}(C') &= D((2.8,1.6), r)
\end{align*}

Using \Cref{lem:thick-corridors-correctness} (and setting $\varepsilon$ small enough as in the argument of \Cref{lem:edge-gadget-tunnel}), we can make the following crucial observation; compare also Figures~\ref{fig:clause-gadget}b--e.
If \configuration enters the main part via $\NegativeEntrance{j}(C')$, the gadget contains two corridors, effectively giving the configuration the flexibility to lean more towards the left or right side of the main part; compare also Figures~\ref{fig:clause-gadget}b and~c.
Since with both corridors we intersect $Q_{\varepsilon}^2$ or $Q_{\varepsilon}^4$, a configuration \configuration that enters the main part $C'$ via $\NegativeEntrance{i}(C')$ \emph{and} $\NegativeEntrance{k}(C')$ cannot be planar if it also uses $\NegativeEntrance{j}(C')$.
However, if \configuration enters the main part via $\PositiveEntrance{i}(C')$ \emph{or} $\PositiveEntrance{k}(C')$, it uses a corridor that does not intersect with the one(s) for $\NegativeEntrance{j}(C')$, %
allowing a planar configuration even if \configuration\ enters the main part via $\NegativeEntrance{j}(C')$. 
The corridor $Q_{\varepsilon}^5$ can always be used.
Since $Q_{\varepsilon}^6$ ($Q_{\varepsilon}^7$) runs completely across $Q_{\varepsilon}^2$ (across $Q_{\varepsilon}^4$), any two configurations that run through these two corridors must cross.
This gives rise to the following observation.
\begin{observation}
	\label{obs:clause-gadget-planar}
	Let $\configuration$ be a planar configuration of $\linkage = (G, \Unit)$ that passes through the clause gadget $C$ for some clause $c \in \mathcal{C}$ and places a vertex in an entrance area of $x_i$, one of $x_j$, and one of $x_k$.
	There is some $v\in V(G)$ such that $\configuration(v) \in \PositiveEntrance{z}(C')$ for some $z \in \{i,j,k\}$.
\end{observation}

\begin{figure}
	\centering
	\includegraphics[page=11]{clause-gadget}
	\caption{\textbf{\textsf{(a)}} The main part of the clause gadget and
    \textbf{\textsf{(b)}}--\textbf{\textsf{(e)}} different configurations \configuration through the main part of the clause gadget.
    }
	\label{fig:clause-gadget}
\end{figure}

To complete the construction, we attach two shifters on the left side of $C'$ and two shifters on the right side of $C'$ to reset the offset that we introduce by using two bends to transport the variable assignment from the variable gadget (which we will introduce next) to the sides of $C'$; recall also \Cref{fig:example}.
Furthermore, we attach two tunnels to the bottom of $C'$.
We unify the entrance and exit areas of the attached edge gadgets and $C'$ and now establish the main properties of the clause gadget.
\begin{lemma}
	\label{lem:clause-gadget}
	Let $C$ be the clause gadget for the clause $(x_i \lor x_j \lor x_k) \in \mathcal{C}$.
	Four of its seven routes have six edges and the remaining three routes have seven edges.
	Furthermore, let $\configuration$ be a configuration of a unit-length linear linkage on $n$ vertices $v_1, \ldots, v_n$ that passes through $C$ starting at $v_q$.
	For $z \in\{i,k\}$, if $\configuration(v_q) \in \NegativeEntrance{z}(C)$
	(respectively $\configuration(v_q) \in \PositiveEntrance{z}(C)$), $q \leq n - 6$, then
	$\configuration(v_{q+6}) \in \NegativeExit{z}(C)$
	(respectively $\configuration(v_{q + 6}) \in \PositiveExit{z}(C)$).
	Furthermore, if $\configuration(v_q) \in \NegativeEntrance{j}(C)$
	(respectively $\configuration(v_q) \in \PositiveEntrance{j}(C)$), $q \leq n - 7$, then
	$\configuration(v_{q+7}) \in \NegativeExit{j}(C)$
	(respectively $\configuration(v_{q + 7}) \in \PositiveExit{j}(C)$).
\end{lemma}
\begin{proof}
	The length of the routes can be seen by considering the main part of the clause gadget depicted in \Cref{fig:clause-gadget} and recalling that we attach shifters/tunnels to its sides; see also \Cref{lem:edge-gadget-shifter,lem:edge-gadget-tunnel}.
	Furthermore, most routes inside the main part satisfy the preconditions for \Cref{lem:thick-corridors-correctness,lem:crossing-strips} for appropriate values of~$\varepsilon$ and $\xi$.
	Observe that for $\varepsilon$ small enough also the intersection between the routes starting at $\NegativeEntrance{j}(C')$ and $\NegativeEntrance{i}(C')$ ($\NegativeEntrance{k}(C')$) satisfies the precondition of \Cref{lem:crossing-strips}.
	Furthermore, note that two routes start at $\NegativeEntrance{j}(C')$ and end at $\NegativeExit{j}(C')$.
	Although these will not satisfy \Cref{lem:crossing-strips}, we can make the following observation.
	If $\configuration$ passes through $C$, then $\configuration(v_{q+2}v_{q+3})$ has to use one of the two corridors.
	This is because the corresponding strips form an angle that is strictly larger than 0.
	Hence, they only overlap close to their common endpoint, and, using similar arguments as in the proof of \Cref{lem:crossing-strips}, the edge $\configuration(v_{q+2}v_{q+3})$ has to lie in one of them.
	The same holds when the two corridors meet at $(2,0)$.
	From there on, we can use \Cref{lem:crossing-strips} to see that also the following edges must stay inside one route.
	By applying \Cref{lem:thick-corridors-correctness}, the lemma statement follows.
\end{proof}

It still remains to describe the gadget for a negative clause.
We obtain it by mirroring the above-described construction along the $x$-axis.
This also changes where the configurations for $x_i$ being true or false arrive.
More concretely, if a variable $x_i$ of a clause $c_j$ is set to true, it now uses the route that it would have previously used when set to false and vice versa.
Consequently, this also affects the definition of the positive and negative entrance areas and, more importantly, allows us to make an observation analogous to \Cref{obs:clause-gadget-planar} for negative clauses.

\paragraph{Variable Gadget.}
The last gadget that we need is the variable gadget.
For the following description, consider a variable $x_i \in \mathcal{X}$ of $\varphi$.
The gadget consists of three main components with different roles (see also \Cref{fig:variable-gadget}): making the configuration \configuration ``set'' the truth assignment $x_i = t$, propagating this to all relevant clauses $c_j \in \mathcal{C}$ with $x_i \in c_j$ or $\lnot x_i \in c_j$, and ``resetting'' \configuration before entering the variable gadget of $x_{i + 1}$, if it exists.
Since the second component is just a series of edge gadgets that connect the variable gadget to the clause gadgets, we focus on the first and third component of the variable gadget.
\medskip

\noindent\textit{First Component.}
Assume that the first of the above-mentioned edge gadgets, let it be $F$, has its lower-left corner at the origin.
We now want to find a starting point $s_i = (x_{s_i},y_{s_i})$ for the configuration such that it can place the next vertex in $\NegativeEntrance{i}(F)$ or $\PositiveEntrance{i}(F)$ (and ideally at the center of these disks).
Consider the intersection point of the unit circles centered at $(0,0)$ and $(0,0.6)$ which lies to the left of $F$.
It is at $y = 0.3$, which already fixes the $y$-coordinate of $s_i$ to be $y_{s_i} = 0.3$.
However, using Pythagoras we observe that the $x$-coordinate of the intersection point is $x = -\sqrt{1 - 0.3^2}$.
One can easily see that $x \notin \mathbb{Q}$ holds, which makes it unsuited for our purposes (since $s = s_1$ and $t = t_N$ must be rational).
However, in the spirit of \Cref{obs:thick-corridors-witness}, we can select a small number $\tau \ll r$ such that $x + \tau$ is rational.
This allows us to fix the point $s_i$ at $(-\sqrt{1 - 0.3^2} + \tau, 0.3)$.
Let $s_i^* = (-\sqrt{1 - 0.3^2}, 0.3)$.
The actual first component of the gadget consists of the two $\varepsilon$-thick corridors $(s_i^*, (0,0))$ and $(s_i^*, (0,0.6))$; see \Cref{fig:variable-gadget}a.
Note that we use here a slight abuse of definition and allow corridors on two vertices.
Since these corridors are concatenated with the respective corridors of the second component, they will in the end have more than three vertices as required.

\begin{figure}
	\centering
	\includegraphics[page=4]{variable-gadget}
	\caption{The \textbf{\textsf{(a)}} first and~\textbf{\textsf{(b)}} third component of the variable gadget for $x_1$ and $x_i$, respectively. The blue cross indicates the point $s_1^*$, which is an intersection point of the two unit circles centered at $(0,0)$ and $(0,0.6)$. The violet squares indicate $s_1$ and $t_i$.}
	\label{fig:variable-gadget}
\end{figure}

\medskip

\noindent\textit{Third Component.}
Now we assume, for ease of presentation, that the last edge gadget of the second component, let it be $F$, has its lower-right corner at the origin.
The third component of the variable gadget, depicted in \Cref{fig:variable-gadget}b, uses an analogous idea to force \configuration to converge ``near'' a single point $t_i^* = (2 - \sqrt{1 - 0.3^2}, 0.3)$.\footnote{Later, we will unify $t_i^*$ and $s_{i + 1}^*$ (if this point exists), which is the case if we place $X_{i + 1}$ such that the lower-left corner of its first edge gadget lies at $x = 2$ (with respect to the coordinates used here).}
The point $t_i$ is then, similar to $s_i$, offset by $\tau$, i.e., $t_i = t_i^* + (\tau, 0)$.
More concretely, it is composed of the two $\varepsilon$-thick corridors defined by the paths $((0,0), (1, 0), t_i' , t_i^*)$ and $((0, 0.6), (1, 0.6), t_i', t_i^*)$ where $t_i'$ is the intersection of the unit circles centered at $t_i^*$ and $(1, 0)$ which is to the left of $t_i^*$.
Note that we would get the same point if we would instead consider the intersection between the circles centered at $t_i^*$ and $(1, 0.6)$ that is to the left of $t_i^*$.
Observe that both corridors meet at $t_i'$ before they run horizontally towards $t_i$.
The configuration will set (and ``reset'') the variable assignment inside the gadget.
Thus, there is only one entrance and exit area for the variable gadget $X_i$:
\begin{align*}
	\NewEntrance{}{i}(X_i) &= D(s_i^*, r)
	&\quad\quad
	\NewExit{}{i}(X_i) &= D(t_i^*, r)
\end{align*}

We are now ready to show correctness of the variable gadget:
\begin{lemma}
	\label{lem:variable-gadget}
	Let $X_i$ be the variable gadget for the variable $x_i \in \mathcal{X}$.
	Each of its two routes has the same number $m_{X_i} = k$ of edges for some $k > 3$.
	Furthermore, let $\configuration$ be a configuration of a unit-length linear linkage on $n$ vertices $v_1, \ldots, v_n$ that passes through $X_i$ starting at $v_j$, $j \leq n - k$.
	If $\configuration(v_j) \in \NewEntrance{}{i}(X_i)$ then
	$\configuration(v_{j+k}) \in \NewExit{}{i}(X_i)$.
	Moreover, for both truth values of $x_i$, there is a planar configuration $\configuration$ of $\linkage$ that passes through $X_i$ such that $\configuration(v_j) = s_i$ and $\configuration(v_{j+k}) = t_i$.
\end{lemma}
\begin{proof}
	The first part of the statement can be shown as follows.
	On the one hand, for the first and third component we can use analogous arguments as for the tunnel (recall \Cref{lem:edge-gadget-tunnel} and see \Cref{fig:variable-gadget} for the edge count of the two routes through these components) and the clause gadget regarding corridors starting or ending at the same vertex (recall \Cref{lem:clause-gadget}).
	On the other hand, for the second component we observe that the number of edge and clause gadgets a route has to pass through is independent of the truth value it represents.
	Since $x_i$ occurs in at least one clause $c_j$, the second component must exist and since it is composed of clause and edge gadgets, correctness follows by combining \Cref{lem:clause-gadget,lem:edge-gadget-bend,lem:edge-gadget-shifter,lem:edge-gadget-tunnel}.

	For the second part of the statement, we now construct a configuration \configuration as follows.
    Consider a truth value $t$ for $x_i$ and let $\mathcal{Q}_t$ be the path from $s_i^*$ to $t_i^*$ for the truth value $t$.
    More concretely, $\mathcal{Q}_t$ starts at~$s_i^*$, continues at $(0,0)$ (if $t$ is false) or $(0,0.6)$ (if $t$ is true) and traverses the relevant edge and clause gadget until it arrives at $t_i^*$.
    Observe that every segment of $\mathcal{Q}_t$ has length one and is contained inside $P$.
    Moreover, the $\varepsilon$-corridor for $\mathcal{Q}_t$ is also contained in $P$, so we can apply \Cref{obs:thick-corridors-witness} to get the sought-after configuration \configuration.
    More concretely, we place $\configuration(v_j) = s_i$.
	The next vertex $v_{j + 1}$ will be placed at $\configuration(v_{j + 1}) = (\tau, 0)$ (if $t$ is false) or at $\configuration(v_{j + 1}) = (\tau, 0.6)$ (if $t$ is true), which is inside $\NegativeEntrance{i}(F)$ or $\PositiveEntrance{i}(F)$, respectively.
    Since the $x$-distance between $s_i^*$ and $s_i$ is $\tau$ and $\Vert s_i^* - (0,0)\Vert_2 = \Vert s_i^* - (0,0.6)\Vert_2 = 1$, this edge has length 1. 
	By \Cref{obs:thick-corridors-witness}, we can extend \configuration such that it witnesses all the edge gadgets and clause gadgets while placing the vertices of \linkage exactly $(\tau,0)$ away from the intended positions defined by the $\varepsilon$-corridors.
	In particular, we have $\configuration(v_{j + k}) = t_i^* + (\tau, 0)$.
    The statement now follows by observing $t_i^* + (\tau, 0) = t_i$.
\end{proof}

\subsection{The Complete Reduction}
\label{sec:hardness-combine}

We now combine the above-introduced gadgets.
To this end, let $\varphi = (\mathcal{X}, \mathcal{C}, \mathcal{E})$ be an instance of \probname{Planar Monotone 3-Sat} together with a planar rectilinear drawing $\mathcal{E}$ of the incidence graph.
Towards constructing $P$, we observe that every gadget from \Cref{sec:hardness-details} has constant size.
Therefore, we can scale~$\mathcal{E}$ by a polynomial factor such that we can replace all components with the respective gadgets without creating overlaps or (unwanted) crossings.
We take the union over all $\varepsilon$-corridors used in the gadgets.
To define $\xi$, $\varepsilon$, and $r$, we consider the set $\mathcal{A}$ of all turning angles between two segments used to define an $\varepsilon$-corridor (even across gadgets).
We let $\xi = \min_{\alpha \in \mathcal{A}} \sin(\alpha) > 0$ (since no two consecutive segments are collinear).
Furthermore, we pick $\varepsilon < 0.1$ to be a rational value such that no two non-adjacent strips from the same $\varepsilon$-corridor intersect and every area (defined using $r = 4\varepsilon / \xi$) only intersects its adjacent strips and no other area unless at shared endpoint of gadgets (except inside bend gadgets, where we have argued in \Cref{lem:edge-gadget-bend} that this is fine).
Moreover, we ensure that $\varepsilon \leq \xi^2/16$ (due to \Cref{lem:thick-corridors-correctness}), $\varepsilon \leq \mu^2\xi/18$ (where~$\mu$ is the global minimum distance from any vertex defining an $\varepsilon$-corridor to any strip not incident to it as in \Cref{lem:crossing-strips}, except inside bends and at the ``forks'' inside clause and variable gadgets, where we have argued correctness separately in \Cref{lem:edge-gadget-bend,lem:variable-gadget,lem:clause-gadget}), and the corridors in the clause gadget intersect exactly as in \Cref{fig:clause-gadget}a.
Observe that $\varepsilon$ is a rational constant (of constant precision) that is independent of $\varphi$ since once $\mathcal{E}$ is scaled sufficiently that neighboring gadgets do not intersect, $\varepsilon$ only depends on the angles inside the (instance-independent) gadgets.
This allows us also to pick the offset $\tau$ from the variable gadget to ensure rational coordinates (of constant precision) for every $s_i$ and $t_i$, $i \in [N]$.

The polygonal domain $P$ is defined as (union of) all points strictly inside some $\varepsilon$-corridors.
By \Cref{obs:thick-corridors-rational}, we can realize every individual $\varepsilon$-corridor using rational coordinates.
To see that $P$ can be encoded in polynomial size, we first observe that there are only five different ``kinds'' of gadget (three edge gadgets, clause gadget, and variable gadget).
The shape of these gadgets does not depend on the input formula $\varphi$.
Hence, every gadget can be represented with rational coordinates (of constant precision).
Each gadget is placed by a combination of the three operations:
First, translating the gadget to the correct position in $\mathcal{E}$, which lies on a grid of polynomial size and bit complexity.
Afterwards, reflecting the gadget or rotating it by $90^\circ$.
All these operations keep the bit complexity of the required coordinates logarithmic in the size of $\mathcal{E}$.
Thus, $P$ can be encoded in polynomial size.
To finish the construction of $(\linkage = (G, \Unit), P, V', \delta)$, we set $V' = \{v_1, v_n\}$ and $\delta(v_1) = s_1 = s$ and $\delta(v_n) = t_N = t$, where $n$ denotes the number of vertices in $G$, which we specify next.
We note that the obtained polygonal domain $P$ contains (polynomially many) obstacles, i.e., holes.
Let $n_T$, $n_B$, and $n_S$ denote the number of tunnels, bends, and shifters used to create $P$, excluding those used in other gadgets such as the clause gadget.
The linear linkage \linkage consists of $n = 2n_T + 7n_B + 2n_S + 4\Size{\mathcal{X}} + 19\Size{\mathcal{C}} + 1$ vertices. %

Recall that our correctness lemmas for the gadgets always assumed that a configuration \configuration of \linkage passes through them.
Now, we show that this must indeed be the case.
Intuitively, every configuration from $s = s_1$ to $t = t_N$ must visit $t_i = s_{i + 1}$ for $1 \leq i < N$.
Since \linkage consists of exactly the required number of vertices to move from $s = s_1$ via $t_i = s_{i + 1}$ for all $1 \leq i < N$ to $t_N$, \configuration cannot ``waste'' vertices by turning around, i.e., it must pass through every variable gadget (and, consequently, also every edge and clause gadget as shown for \Cref{lem:variable-gadget}).
More formally, we show the following:
\begin{lemma}
	\label{lem:passes-through}
	Let $\configuration$ be a configuration of $\linkage$ that lives inside $P$ with
	$\configuration(v_1) = s$ and $\configuration(v_n) = t$. 
	Then $\configuration$ passes through every gadget $F$ that it enters, i.e., if $\configuration(v_i)$ lies in an entrance area of $F$, then $\configuration$ passes through $F$ starting at $v_i$.
\end{lemma}
\begin{proof}
	We first argue that \configuration must place each of its vertices in an area by induction on the vertices $v_k$ of $\linkage$.
	The base case, i.e., for the first vertex $v_1$, follows by the definition of \configuration.
	For larger $k$, let $\configuration(v_k)$ lie in the area of a vertex $q$ of a route for some $k < n$.
	Moreover, let $Q_\varepsilon$ be the $\varepsilon$-thick corridor of a route through~$q$.
	Hence, the segment $\sigma = \configuration(v_kv_{k+1})$ lies in $Q_\varepsilon$ (and possibly in strips of intersecting routes).
	However, by our choice of $\varepsilon$, we ensure that $\sigma \subseteq Q_{\varepsilon}$; recall \Cref{lem:crossing-strips} (and the separate arguments for bends, clause and variable gadgets in \Cref{lem:clause-gadget,lem:edge-gadget-bend,lem:variable-gadget}, respectively).
	Now, we can apply \Cref{lem:thick-corridors-unit-segment} on $\sigma$ and $Q_{\varepsilon}$ to conclude that $\configuration(v_{k+1})$ must lie in the area of a neighbor $p$ of $q$ along this route.

	Consider the paths from $s_1^*$ to $t_N^*$ along all possible routes, and let $q$ be a vertex of one of these paths. Moreover, consider a path from $s_1^*$ to $q$ along a route.
	For every gadget, we have shown that every route through that gadget consists of the same number of edges (for the clause gadget, it has to belong to the same variable); recall \Cref{lem:clause-gadget,lem:edge-gadget-bend,lem:edge-gadget-shifter,lem:edge-gadget-tunnel,lem:variable-gadget}.
	Consequently, each of the above considered paths has the same number of edges $m_q$. 
	Moreover, if two vertices $p$ and $q$ are adjacent on a route, then $\Size{m_p - m_q} = 1$. 
	Finally, we have $m_{s_1^*} = 0$ and $m_{t_N^*} = n - 1$.
	
	By our choice of $\tau$, $\configuration(v_1) = s_1$ lies only in the area of $s_1^*$, and $\configuration(v_n) = t_N$ lies only in the area of $t_N^*$.
	The values $m_q$ of the areas containing
	$\configuration(v_1), \ldots, \configuration(v_n)$ change by exactly one in each of the $n - 1$ steps. As they start at $0$ and end at $n - 1$, they can only increase in every step.
	Thus, $\configuration$ follows one of the above paths, i.e., $\configuration(v_k)$ lies in the area of a vertex $q$ with $m_q = k - 1$, for every $k \in [n]$.
	
	Finally, assume that $\configuration(v_i)$ lies in an entrance area of a gadget $F$ defined at some vertex $q$.
	Then $\configuration(v_{i+1})$ lies in the area of a neighbor $p$ of $q$ with $m_p = i$, i.e., of the next vertex of a route of $F$.
	In particular, $\configuration(v_{i+1})$ lies in an area of $F$. 
	The exit areas of $F$ have
	an $m$-value of $i - 1 + m_F \leq m_{t_N^*} = n - 1$, so $i + m_F \leq n$. 
	Since the $m$-values are strictly increasing, no area of $F$ contains two vertices. 
	This allows us to conclude that $\configuration$
	passes through $F$ starting at $v_i$.
	Since $F$ and \configuration were selected arbitrarily, this holds for every configuration and gadget.
\end{proof}

Finally, using \Cref{obs:clause-gadget-planar}, we conclude that for~\configuration to be planar, %
every clause must be satisfied, %
establishing \NP-hardness of \realizabilityShort:
\begin{restatable}\restateref{thm:hardness}{theorem}{theoremHardness}
	\label{thm:hardness}
	\realizability remains \NP-hard even if $\linkage = (G, \Unit)$ is a unit-length linear linkage where only the first and last vertex are constrained.
\end{restatable}
\begin{prooflater}{ptheoremHardness}
	Let $\varphi = (\mathcal{X}, \mathcal{C}, \mathcal{E})$ be an instance of \probname{Planar Monotone 3-Sat} consisting of $N$ variables $\mathcal{X} = \{x_1, \ldots, x_N\}$ and $M$ clauses $\mathcal{C} = \{c_1, \ldots, c_M\}$ together with a planar rectilinear drawing $\mathcal{E}$ of the incidence graph.
	Let $(\linkage = (G, \Unit), P, V', \delta)$ be the above-defined instance of \realizabilityShort.
	Note that for every two variables $x_i, x_{i + 1} \in \mathcal{X}$ with $1 \leq i < N$ we have $t_{i} = s_{i + 1}$.
	By \Cref{obs:thick-corridors-rational}, $P$ has rational coordinates and the size of the instance can be bounded by a polynomial in $N + M$.
	Moreover, since a sufficiently-good approximation of the irrational coordinates used to define the $\varepsilon$-corridors in the construction suffices, we can compute $(\linkage, P, V', \delta)$ in polynomial time.
	It remains to show correctness of the reduction.
	
	\paragraph*{($\boldsymbol{\Rightarrow}$)}
	Let $\varphi$ be a positive instance of \probname{Planar Monotone 3-Sat} and $\Psi\colon \mathcal{X} \to \{0,1\}$ a satisfying assignment.
	We now construct based on $\Psi$ a configuration $\configuration$ of \linkage that witnesses that $(\linkage, P, V', \delta)$ is a positive instance of \realizabilityShort.
	To that end, we let $\configuration(v_1) = s_1 = s$ and consider for the placement of $v_2$ the truth assignment of $x_1$.
	If $\Psi(x_1) = 1$, we place $v_2$ such that $\configuration(v_2) \in \PositiveEntrance{1}(T)$ holds, where $T$ is the tunnel attached to the first component of the variable gadget for $x_1$. 
	Otherwise, i.e., if $\Psi(x_1) = 0$, we place $v_2$ such that $\configuration(v_2) \in \NegativeEntrance{1}(T)$ holds.
	We then follow the construction indicated in the proof of \Cref{lem:variable-gadget}.
	If we reach the gadget $C$ of a clause $c_j \in \mathcal{C}$ with $x_1 \in c_j$ (or $\lnot x_1 \in c_j$), we decide based on the truth state of $x_1$ (and the other literals of $c_j$) which corridor we have to use.
	More concretely, if~$x_1$ is the first or third variable of $c_j$, the truth value of $x_1$ completely determines the corridor to use.
	The same holds true if $x_1$ is the second variable of $c_j$ and it satisfies $c_j$.
	Otherwise, we decide based on the truth state of the other two variables in $c_j$ to which side the configuration \configuration should lean.
	Since $\Psi$ satisfies~$\varphi$, it in particular satisfies $c_j$.
	Thus, there is some other variable $x_q \in c_j$ that satisfies $c_j$.
	We can lean \configuration towards the side of $x_q$ to avoid a crossing in $\configuration$; see the different cases in \Cref{fig:clause-gadget}b--e.
	We continue creating \configuration until we reach the end of the variable gadget.
	Using \Cref{lem:variable-gadget}, we are able to reach the point $t_1$.
	From there on, we continue in the same manner for the remaining variables in $\varphi$ until we eventually reach $t_N = t = \delta(v_n)$.
	Since configurations in other gadgets do not cross, we can conclude that the constructed configuration $\configuration$ witnesses that $(\linkage, P, V', \delta)$ is a positive instance of \realizabilityShort.
	
	\paragraph*{($\boldsymbol{\Leftarrow}$)}
	Let $(\linkage, P, V', \delta)$ be a positive instance of \realizabilityShort and $\configuration$ a witnessing configuration which must, by \Cref{lem:passes-through}, pass through every gadget it enters.
	Based on $\configuration$, we now construct a truth assignment $\Psi\colon \mathcal{X} \to \{0,1\}$ that satisfies $\varphi$.
	We start with $x_1$ and look at the configuration inside its variable gadget, in particular at its first component.
	Since $\configuration(v_1) = s \in \NewEntrance{}{1}(X_1)$, we can conclude that the second vertex $v_2 \in V(G)$ has $\configuration(v_2) \in \NegativeEntrance{1}(F_1) \cup \PositiveEntrance{1}(F_1)$, where $F_1$ is the first gadget of the second component of the variable gadget $X_1$.
	Note that the disks $\NegativeEntrance{1}(F_1)$ and $\PositiveEntrance{1}(F_1)$ are disjoint.
	If we have $\configuration(v_2) \in \NegativeEntrance{1}(F_1)$, we set $\Psi(x_1) = 0$, otherwise, i.e., if we have $\configuration(v_2) \in \PositiveEntrance{1}(F_1)$, we set $\Psi(x_1) = 1$.
	
	By \Cref{lem:variable-gadget} we will, by following the configuration $\configuration$, eventually arrive at a vertex $v' \in V$ such that $\configuration(v') \in\NewExit{}{1}(X_1)=\NewEntrance{}{2}(X_2)$.
	From there on, we then proceed with constructing the truth assignment~$\Psi$.
	By applying \Cref{lem:variable-gadget} iteratively, $\Psi$ will eventually assign a truth value to all variables of~$\mathcal{X}$.
	
	It remains to show that $\Psi$ is a satisfying assignment of $\varphi$.
	Towards a contradiction, assume that $\Psi$ would not satisfy a clause $c_q \in \mathcal{C}$.
	Assume that $c_q = x_i \lor x_j \lor x_k$ (the case where $c_q$ is a negative clause can be handled by analogous arguments; recall that then the roles of the routes are swapped).
	As $\Psi$ does not satisfy $c_q$, we must have set $\Psi(x_i) = \Psi(x_j) = \Psi(x_k) = 0$.
	Due to the definition of $\Psi$ and since truth assignments are preserved along the edge gadgets (recall \Cref{lem:edge-gadget-bend,lem:edge-gadget-shifter,lem:edge-gadget-tunnel}), this implies that
	there are three vertices $v_{i'}, v_{j'}, v_{k'} \in V(G)$ such that we have, without loss of generality, $\configuration(v_{i'}) \in \NegativeEntrance{i}(C')$, $\configuration(v_{j'}) \in \NegativeEntrance{j}(C')$, and $\configuration(v_{k'}) \in \NegativeEntrance{k}(C')$, where $C'$ is the main part of the clause gadget $C$ for~$c_q$.
	However, \Cref{obs:clause-gadget-planar} tells us that this is not possible in a planar configuration.
	Thus, we must have $\configuration(v_{i'}) \in \PositiveEntrance{i}(C')$, $\configuration(v_{j'}) \in \PositiveEntrance{j}(C')$, or $\configuration(v_{k'}) \in \PositiveEntrance{k}(C')$, i.e., $\Psi(x_i) = 1$, $\Psi(x_j) = 1$, or $\Psi(x_k) = 1$.
	Therefore, we conclude that $\Psi$ satisfies all clauses of $\varphi$, and in particular $c_q$, i.e., $\varphi$ is a positive instance of \probname{Planar Monotone 3-Sat}.
\end{prooflater}

\newcommand{\stateThreeEdgesTheorem}{%
\begin{restatable}\restateref{thm:three-edges}{theorem}{theoremThreeEdges}
	\label{thm:three-edges}
	\ullRealizability for three-edge linkages %
	in a %
	convex polygon~%
	$P$ with $m$ vertices can be solved in time \BigO{m} %
	even for general linear linkages.
\end{restatable}
}%

\section{Linear Linkages of Length Three in Convex Polygons}
\label{sec:three-edge-linkages}
\Cref{thm:hardness} raises the question in which settings we can solve \realizabilityShort for linear linkages with constrained endpoints efficiently.
Observe that \NewText{our reduction for} \Cref{thm:hardness} %
\NewText{relies} on $P$ being a polygonal domain with a polynomial number of holes.
As our last result, we show that if $P$ is convex and $G$ %
\NewText{consists} of three edges, we can solve \realizabilityShort in linear time.

\begin{figure}
	\centering
	\includegraphics{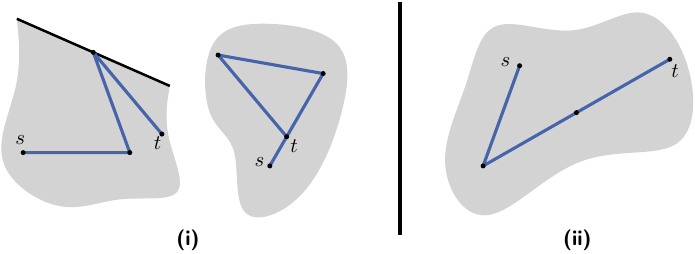}
	\caption{Realizing a three-edge linkage. It suffices to check these types of configurations; see~\Cref{thm:linear-three-edges}. The polygonal domain is %
		\NewText{indicated} by the gray areas.}
	\label{fig:3-link}
\end{figure}
\begin{restatable}\restateref{thm:linear-three-edges}{theorem}{theoremThreeEdges}
	\label{thm:linear-three-edges}
	\realizability for three-edge linear linkages $G$ in a convex polygon $P$  where we constrain the first and last vertex of $G$ can be solved in $\BigO{n_P \log n_P}$ time even for %
	arbitrary~$\ell$.
\end{restatable}
\begin{proofsketch}
    For this sketch, we allow degenerate configurations where a point might lie on the boundary of $P$ or on an edge of the linkage. 
    We discuss in the full proof how to deal with those.
    
	Assume there exists a non-crossing configuration $\configuration$ of the linkage with vertices $v_1,\ldots, v_4$ in $P$.
    Let $s = \delta(v_1)$ and $t = \delta(v_4)$.
	Observe that any such configuration has at most one degree of freedom (1-dof).
	We now aim to reconfigure~$\configuration$ until either (i) $v_2$ or $v_3$ hits the boundary of $P$, (ii) two incident edges become collinear, or (iii) $v_1$ lies on $\configuration(v_3v_4)$ or $v_4$ lies on $\configuration(v_1v_2)$; see~\Cref{fig:3-link}. 
	If this is not possible, we have to be in case (i)--(iii) already.
	Observe that if the configuration space (allowing self-crossings) is not restricted by $P$, then we can realize the linkage such that it forms with the edge $\configuration(s)\configuration(t)$ a convex quadrilateral.
	This will either increase the angle at $\configuration(v_2)$ or $\configuration(v_3)$ \citep[see][Lemma 5.3.2]{DO.GFA.2007} and hence leads to case (ii).
	
	Consequently, if we start moving the linkage from configuration $\configuration$ we will either introduce a crossing (case (iii)), hit the boundary (and since $P$ is convex this happens at some vertex, thus at $v_2$ or $v_3$, which is covered by case (i)), or two neighboring edges become collinear (case (ii)).
	So it suffices to check all configurations for which one of the cases holds.

	The number of configurations is in $\BigO{m_P}$, and in each case we have to solve an instance of \realizabilityShort for a two-edge linkage, which takes constant time. 
    Since $m_P \in \BigO{n_P}$, the theorem follows.
\end{proofsketch}
\begin{prooflater}{ptheoremThreeEdges}
Assume that there exists a noncrossing configuration of the
linkage in $P$. We name its vertices $v_1,\ldots,v_4$, with
$\configuration(v_1)=s$ and $\configuration(v_4)=t$.
Let $a,b,c$ denote the three edge lengths, respectively,
and let $d=\|t-s\|$. The idea of the proof is, that we can transform 
every configuration into a degenerate configuration, from which there
are only $\BigO{n_P}$ many relevant candidates. We go through all candidate configurations
and check if they might have been created as a limit of a (valid) configuration, in 
which case the linkage would have been realizable. The situations where points coincides, i.e., when
$a = b$ and  $c = d$, or when $a = d$ and $b = c$ are special and we will discuss 
them in the end of the proof. We first enumerate noncoincident candidates for every length assignment, including the exceptional assignments. Coincident-vertex families are handled afterwards.

For the purpose of identifying degenerate configurations, we
allow contacts with the boundary and between nonincident
parts of the linkage. We additionally regard $s$ and $t$
as punctures. Throughout, $P$
denotes the original convex polygon.

Starting from a valid configuration, we seek a continuous
motion leading, possibly in the limit, to one of the
following events:
\begin{enumerate}
    \item[(i)] A movable vertex reaches the original boundary
    of $P$, or a nonincident part of the linkage touches
    one of the punctures $s$ and $t$.
    \item[(ii)] Two incident edges become collinear.
\end{enumerate}
Contact with the original polygon boundary occurs at a
vertex since $P$ is convex. Contact with a puncture may instead
occur in the interior of a nonincident edge.

We first justify that such an event can always be reached.
As long as neither pair of adjacent edges is collinear, 
a sufficiently small movement of either movable vertex 
along its circle uniquely determines a continuous movement 
of the other vertex. Thus, the linkage locally follows 
a single motion branch, which can be parametrized by the 
position of either movable vertex.
Suppose that, starting from a valid configuration, no event 
could be reached or approached. We could then continue 
moving either movable vertex around its circle in a fixed direction. 
This, both $v_2$ and $v_3$ could each complete a full turn 
around their pinned neighbors.
Assume that $v_2$ makes a full turn. If $a\ge d$ we would 
clearly caused event (i), so $a<d$. We rotate now such that
$v_2$ lies inside the segment $st$. The triangle formed by $v_2,v_3,t$
gives then $b\le(d-a)+c$ (equality would imply event (ii))
and thus $a+b<d+c$. 
Now rotate $v_3$ such it is collinear with $s$ and $t$ such that it lies beyond $t$. 
We get that by the triangle
inequality ($s,v_2,v_3$) and the absence of (ii) that 
$c+d<a+b$, which yields the contradiction.

We next bound the number of degenerates
configurations. 
For contacts with the original polygon boundary,
$\configuration(v_2)$ lies at an intersection of the circle
of radius $a$ centered at $s$ with a polygon edge, or
$\configuration(v_3)$ lies at an intersection of the circle
of radius $c$ centered at $t$ with a polygon edge.
There are $\BigO{n_P}$ such positions. Each determines at most
two positions of the other movable vertex by intersecting
two circles. Their centers are distinct, since the pins
lie strictly inside the original polygon.

Contacts between a puncture and the middle edge are already covered by case (ii). 
For the remaining puncture contacts. For a puncture contact 
without coincident vertices,
the endpoint of the edge touching the puncture lies on
the line through $s$ and $t$. Its fixed distance from its
adjacent pin leaves at most two possible positions.
The other movable vertex is again determined by the
intersection of two circles with distinct centers.
There are therefore only constantly many such candidates.

Finally, consider case (ii) without coincident vertices.
If the first two edges are collinear, the distance of
$\configuration(v_3)$ from $s$ belongs to
$\{a+b,|a-b|\}$. The zero-distance case will be handled later. 
Intersecting the corresponding circles
centered at $s$ with the circle of radius $c$ centered
at $t$ gives constantly many candidates. Collinearity
of the last two edges is handled symmetrically.

Altogether, there are $\BigO{n_P}$ candidate limiting
configurations. A candidate that already satisfies
the original embedding and containment conditions is
accepted directly. For every other candidate, we must
check if it could be a limit of a motion. We precompute for 
$v_2$ the angular space at $s$ where $v_2$ lies in the interior 
on $P$ as an ordered set of open intervals. This can be done edge 
by edge in $\BigO{n_P\log n_P}$ time. We do the same for $v_3$'s movement relative to $t$.
Now for each continuation of the degenerate linkage realization, 
determine the orientation of angular movement at $s$ and $t$.
Also test the embedding conditions along that continuation. 
These are constant-size algebraic sign tests because the linkage has fixed size. Consequently, each candidate requires a constant number of oracle queries.
We can use a binary search over the precomputed open intervals to see if $v_2$
and $v_3$ move inside $P$. This takes $\BigO{\log n_P}$ time per degenerate configuration and
therefor $\BigO{n_P \log n_P}$ time in total.

It remains to settle the case where vertices might coincide.
The problem is for example 
that by setting $t=\configuration(v_2)$ we can rotate $v_3$ around
$t$ and obtain an infinite number of degenerate configurations. 
However, we will show that only two of them can be obtained 
as a limit of a motion of valid configuration. 

All valid configurations in this exceptional case have a 
special property. They can be obtained by reflecting $t$ across the line through $s$ and $v_3$
to obtain the position of $v_2$.
Reflection preserves the required edge lengths, and
noncollinear members of this family are noncrossing, since
the first and last edges lie on opposite sides of the
reflection axis. 
As the vertices are distinct, the line through
$s$ and $v_3$ is therefore the perpendicular bisector
of the segment joining $v_2$ and $t$. See also~\Cref{fig:reflection}. 

Ignoring the containment in $P$ 
there are only two limits of valid motions, in which vertices coincide with $t=\configuration(v_2)$. 
Consequently, only two additional candidates arise from the coincidence $v_2=t$: the positions where $v_3$ lies on both $C$ and the line through $s$ and $t$. We test all local continuations at these candidates as before. The coincidence $v_3=s$ is handled symmetrically. Thus, coincident vertices contribute at most four additional candidates, and the candidate set remains of size $\BigO{n_P}$.
\begin{figure}
    \centering
    \includegraphics{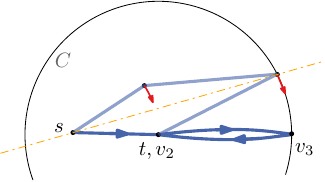}
    \caption{The case $a=d$ and $b=c$ and one of the possibilities how collinearities can occur.}
    \label{fig:reflection}
\end{figure}
\end{prooflater}

\section{Concluding Remarks}
\label{sec:conclusion}
We see this paper as a further step towards understanding the complexity of realizing linkages.
Our paper shows that this problem is surprisingly hard in the presence of a polygonal domain even for relatively simple linkages.
The \W[1]-hardness from \Cref{thm:w1} underlines that, from a parameterized complexity perspective, %
\NewText{one has to take parameters of $P$ into account, for example its number of holes or concave corners.}
We see this as an interesting direction for future work.
Our \NP-hardness from \Cref{thm:hardness} raises the question in which settings we can solve \realizabilityShort for linear linkages where we constrain the first and last vertex of $G$ in polynomial time.
\Cref{thm:linear-three-edges} already provides a partial answer to this question.
However, \Cref{thm:linear-three-edges} hinges on three-edge linkages having only one degree of freedom and an extension to four-edge linkages is highly non-trivial.
Therefore, we see a generalization to arbitrary constant-size linear linkages (that stems from insights into the specific problem-variant and is thus more efficient %
\NewText{than} our \XP-algorithm from \Cref{thm:xp}) as a natural next step.
{%
The holes in the polygonal domain $P$ were an important ingredient of the \NP-hardness constructions behind \Cref{thm:hardness}, and we cannot remove them without fundamentally changing the construction.
Therefore, showing hardness for simple polygons $P$ requires new approaches (as has been the case for the \probname{Two-Dimensional Packing} problem \citep[see][]{AMS.FCT.2024}).
Hence, we see establishing \NP-hardness for simple/convex $P$ and $\ExistsR$-hardness for general~$P$ as interesting directions to pursue for linear linkages.
}%
Finally, extensions to other length-constrained drawings, such as dichotomous drawings, as studied for example by \citet{Angelini2025}, are also worth considering.

\acknowledgements
\label{sec:ack}
This work started at the 19th European Research Week on Geometric Graphs (GGWeek'24) in Trier. We thank the organizers for the fruitful and nice atmosphere.
A preliminary version of this paper was presented at the 41st European Workshop on Computational Geometry (EuroCG'25) and also appeared in the proceedings of the 37th International Workshop on Combinatorial Algorithms (IWOCA'26)~\citep{IWOCA}.

\subsection*{Statement on the Use of AI Tools}
The authors acknowledge the use of AI Tools (in particular Claude and ChatGPT) when preparing the article for proofreading and correctness verification.
Moreover, AI Tools assisted with the calculations behind \Cref{lem:segment-intersection,lem:crossing-strips,lem:thick-corridors-unit-segment} and streamlined the argumentation in~\Cref{thm:linear-three-edges}.
No such tools were used for the conference version of this article~\citep{IWOCA}.
The authors checked the correctness of the content and take full responsibility for the article.

\bibliographystyle{abbrvnat}
\bibliography{references}

@Book{DO.GFA.2007,
  author    = {Demaine, Erik D. and O’Rourke, Joseph},
  publisher = {Cambridge University Press},
  title     = {Geometric {F}olding {A}lgorithms: {L}inkages, {O}rigami, {P}olyhedra},
  year      = {2007},
  isbn      = {9780511735172},
  doi       = {10.1017/cbo9780511735172},
}

@InProceedings{CDR.PEG.2004,
  author    = {Cabello, Sergio and Demaine, Erik D. and Rote, Günter},
  booktitle = {Proc. 11th International Symposium on Graph Drawing and Network Visualization (GD'03)},
  title     = {{P}lanar {E}mbeddings of {G}raphs with {S}pecified {E}dge {L}engths},
  year      = {2004},
  editor    = {Giuseppe Liotta},
  pages     = {283--294},
  publisher = {Springer},
  series    = {Lecture Notes in Computer Science},
  volume    = {2912},
  doi       = {10.1007/978-3-540-24595-7_26},
}

@Article{CDR.PEG.2007,
  author  = {Cabello, Sergio and Demaine, Erik D. and Rote, Günter},
  journal = {Journal of Graph Algorithms and Applications},
  title   = {{P}lanar {E}mbeddings of {G}raphs with {S}pecified {E}dge {L}engths},
  year    = {2007},
  number  = {1},
  pages   = {259--276},
  volume  = {11},
  doi     = {10.7155/jgaa.00145},
}

@Article{EW.Fel.1990,
  author  = {Eades, Peter and Wormald, Nicholas C.},
  journal = {Discrete Applied Mathematics},
  title   = {Fixed edge-length graph drawing is {NP}-hard},
  year    = {1990},
  number  = {2},
  pages   = {111--134},
  volume  = {28},
  doi     = {10.1016/0166-218x(90)90110-x},
}

@Article{Kem.GMd.1875,
  author  = {Kempe, Alfred B.},
  journal = {Proceedings of the London Mathematical Society},
  title   = {On a General Method of describing Plane Curves of the nth degree by Linkwork},
  year    = {1875},
  number  = {1},
  pages   = {213--216},
  volume  = {s1-7},
  doi     = {10.1112/plms/s1-7.1.213},
}

@Article{KM.Utc.2002,
  author  = {Kapovich, Michael and Millson, John J.},
  journal = {Topology},
  title   = {Universality theorems for configuration spaces of planar linkages},
  year    = {2002},
  number  = {6},
  pages   = {1051--1107},
  volume  = {41},
  doi     = {10.1016/s0040-9383(01)00034-9},
}

@Article{HJW.MP2.1984,
  author  = {Hopcroft, John and Joseph, Deborah and Whitesides, Sue},
  journal = {SIAM Journal on Computing},
  title   = {Movement {P}roblems for 2-{D}imensional {L}inkages},
  year    = {1984},
  number  = {3},
  pages   = {610--629},
  volume  = {13},
  doi     = {10.1137/0213038},
}

@InProceedings{ADD+.WNC.2016,
  author    = {Abel, Zachary and Demaine, Erik D. and Demaine, Martin L. and Eisenstat, Sarah and Lynch, Jayson and Schardl, Tao B.},
  booktitle = {Proc. 32nd International Symposium on Computational Geometry (SoCG'16)},
  title     = {Who {N}eeds {C}rossings? {H}ardness of {P}lane {G}raph {R}igidity},
  year      = {2016},
  editor    = {S{\'{a}}ndor P. Fekete and Anna Lubiw},
  pages     = {3:1--3:15},
  publisher = {Schloss Dagstuhl – Leibniz-Zentrum für Informatik},
  series    = {LIPIcs},
  volume    = {51},
  doi       = {10.4230/LIPICS.SOCG.2016.3},
}

@Article{CDR.SPA.2003,
  author    = {Robert Connelly and Erik D. Demaine and G{\"{u}}nter Rote},
  journal   = {Discrete \& Computational Geometry},
  title     = {Straightening {P}olygonal {A}rcs and {C}onvexifying {P}olygonal {C}ycles},
  year      = {2003},
  number    = {2},
  pages     = {205--239},
  volume    = {30},
  doi       = {10.1007/S00454-003-0006-7},
}

@Article{BDD+.nrt.2002,
  author  = {Biedl, Therese and Demaine, Erik and Demaine, Martin and Lazard, Sylvain and Lubiw, Anna and O’Rourke, Joseph and Robbins, Steve and Streinu, Ileana and Toussaint, Godfried and Whitesides, Sue},
  journal = {Discrete Applied Mathematics},
  title   = {A note on reconfiguring tree linkages: trees can lock},
  year    = {2002},
  number  = {1–3},
  pages   = {293--297},
  volume  = {117},
  doi     = {10.1016/s0166-218x(01)00229-3},
}

@InProceedings{AKRW.cu.2003,
  author    = {Alt, Helmut and Knauer, Christian and Rote, Günter and Whitesides, Sue},
  booktitle = {Proc. 19th International Symposium on Computational Geometry (SoCG'03)},
  title     = {The complexity of (un)folding},
  year      = {2003},
  editor    = {Steven Fortune},
  pages     = {164--170},
  publisher = {ACM},
  doi       = {10.1145/777792.777818},
}

@Article{Kan.Msl.1992,
  author  = {Kantabutra, Vitit},
  journal = {Discrete \& Computational Geometry},
  title   = {Motions of a short-linked robot arm in a square},
  year    = {1992},
  number  = {1},
  pages   = {69--76},
  volume  = {7},
  doi     = {10.1007/bf02187825},
}

@InProceedings{WP.RCE.1996,
  author    = {Sue Whitesides and Naixun Pei},
  booktitle = {Proc. 2nd International Computing and Combinatorics Conference (COCOON'96)},
  title     = {On the {R}econfiguration of {C}hains ({E}xtended {A}bstract)},
  year      = {1996},
  editor    = {Jin{-}yi Cai and C. K. Wong},
  pages     = {381--390},
  publisher = {Springer},
  series    = {Lecture Notes in Computer Science},
  volume    = {1090},
  doi       = {10.1007/3-540-61332-3_172},
}

@Article{DBK.OBS.2012,
  author  = {Mark de Berg and Amirali Khosravi},
  journal = {International Journal of Computational Geometry and Applications},
  title   = {{O}ptimal {B}inary {S}pace {P}artitions for {S}egments in the {P}lane},
  year    = {2012},
  number  = {03},
  pages   = {187--205},
  volume  = {22},
  doi     = {10.1142/s0218195912500045},
}

@InProceedings{JP.CRM.1985,
  author    = {Joseph, Deborah A. and Plantings, W. Harry},
  booktitle = {Proc. 1st International Symposium on Computational Geometry (SoCG)},
  title     = {On the {C}omplexity of {R}eachability and {M}otion {P}lanning {Q}uestions ({E}xtended {A}bstract)},
  year      = {1985},
  pages     = {62--66},
  publisher = {ACM},
  doi       = {10.1145/323233.323242},
}

@Article{HJW.MRA.1985,
  author  = {Hopcroft, John and Joseph, Deborah and Whitesides, Sue},
  journal = {SIAM Journal on Computing},
  title   = {{O}n the {M}ovement of {R}obot {A}rms in 2-{D}imensional {B}ounded {R}egions},
  year    = {1985},
  number  = {2},
  pages   = {315--333},
  volume  = {14},
  doi     = {10.1137/0214025},
}

@Article{GJ.SSP.1988,
  author  = {Grigoriev, Dima and Vorobjov, Jr., Nicolai N.},
  journal = {Journal of Symbolic Computation},
  title   = {Solving Systems of Polynomial Inequalities in Subexponential Time},
  year    = {1988},
  number  = {1–2},
  pages   = {37--64},
  volume  = {5},
  doi     = {10.1016/s0747-7171(88)80005-1},
}

@Book{dBCvKO.CGA.2008,
  author    = {de Berg, Mark and Cheong, Otfried and van Kreveld, Marc and Overmars, Mark},
  publisher = {Springer},
  title     = {Computational Geometry: Algorithms and Applications},
  year      = {2008},
  edition   = {Third},
  doi       = {10.1007/978-3-540-77974-2},
}

@Article{ADD+.Wnc.2025,
  author  = {Abel, Zachary and Demaine, Erik D. and Demaine, Martin L. and Eisenstat, Sarah and Lynch, Jayson and Schardl, Tao B.},
  journal = {Journal of Computational Geometry},
  title   = {Who needs crossings?: {N}oncrossing linkages are universal, and deciding (global) rigidity is hard},
  year    = {2025},
  number  = {1},
  pages   = {378--452},
  volume  = {16},
  doi     = {10.20382/JOCG.V16I1A12},
}

@Book{Cygan2015,
  author    = {Marek Cygan and Fedor V. Fomin and Lukasz Kowalik and Daniel Lokshtanov and D{\'{a}}niel Marx and Marcin Pilipczuk and Michal Pilipczuk and Saket Saurabh},
  publisher = {Springer},
  title     = {Parameterized {A}lgorithms},
  year      = {2015},
  isbn      = {978-3-319-21274-6},
  doi       = {10.1007/978-3-319-21275-3},
}

@Article{EvdHM.SGN.2024,
  author  = {Jeff Erickson and Ivor {van der Hoog} and Tillmann Miltzow},
  journal = {SIAM Journal on Computing},
  title   = {Smoothing the Gap Between {NP} and {ER}},
  year    = {2024},
  number  = {6},
  pages   = {S20--102},
  volume  = {53},
  doi     = {10.1137/20M1385287},
}

@Book{PS.CGI.1985,
  author    = {Franco P. Preparata and Michael Ian Shamos},
  publisher = {Springer},
  title     = {Computational Geometry - {A}n Introduction},
  year      = {1985},
  doi       = {10.1007/978-1-4612-1098-6},
}

@InProceedings{Angelini2025,
  author    = {Angelini, Patrizio and Cornelsen, Sabine and Haase, Carolina and Hoffmann, Michael and Katsanou, Eleni and Montecchiani, Fabrizio and Steiner, Raphael and Symvonis, Antonios},
  booktitle = {Proc. 41st International Symposium on Computational Geometry (SoCG'25)},
  title     = {Geometric Realizations of Dichotomous Ordinal Graphs},
  year      = {2025},
  editor    = {Oswin Aichholzer and Haitao Wang},
  pages     = {9:1--9:16},
  publisher = {Schloss Dagstuhl – Leibniz-Zentrum für Informatik},
  series    = {LIPIcs},
  volume    = {332},
  doi       = {10.4230/LIPICS.SOCG.2025.9},
}

@InProceedings{LMM.CDG.2018,
  author    = {Anna Lubiw and Tillmann Miltzow and Debajyoti Mondal},
  booktitle = {Proc. 26th International Symposium on Graph Drawing and Network Visualization (GD'18)},
  title     = {The Complexity of Drawing a Graph in a Polygonal Region},
  year      = {2018},
  editor    = {Therese Biedl and Andreas Kerren},
  pages     = {387--401},
  publisher = {Springer},
  series    = {Lecture Notes in Computer Science},
  volume    = {11282},
  doi       = {10.1007/978-3-030-04414-5_28},
}

@Article{LMM.CDG.2022,
  author  = {Anna Lubiw and Tillmann Miltzow and Debajyoti Mondal},
  journal = {Journal of Graph Algorithms and Applications},
  title   = {The Complexity of Drawing a Graph in a Polygonal Region},
  year    = {2022},
  number  = {4},
  pages   = {421--446},
  volume  = {26},
  doi     = {10.7155/jgaa.00602},
}

@article{AMS.FCT.2024,
  author       = {Mikkel Abrahamsen and
                  Tillmann Miltzow and
                  Nadja Seiferth},
  title        = {Framework for {\(\exists\)} {\(\mathbb{R}\)}-Completeness of Two-Dimensional
                  Packing Problems},
  journal      = {TheoretiCS},
  volume       = {3},
  year         = {2024},
  doi          = {10.46298/THEORETICS.24.11},
}

@InProceedings{IWOCA,
  author    = {Thomas Depian and Carolina Haase and Martin N{\"{o}}llenburg and Andr{\'{e}} Schulz},
  editor    = {Florent Foucaud and Aline Parreau},
  title     = {Realizing Planar Linkages in Polygonal Domains},
  booktitle = {Proc. 37th International Workshop on Combinatorial Algorithms (IWOCA'26)},
  series    = {Lecture Notes in Computer Science},
  volume    = {16587},
  pages     = {251--265},
  publisher = {Springer},
  year      = {2026},
  doi       = {10.1007/978-3-032-27732-9\_18}
}

\end{document}